\documentclass{article}

\usepackage{setspace}
\usepackage{color,soul}
\usepackage{amsfonts}
\usepackage{dsfont}
\usepackage[cmex10]{amsmath}
\usepackage{amssymb}
\usepackage{multirow}

\usepackage[table]{xcolor}

\usepackage{algorithm}%
\usepackage{algpseudocode}%
\algnewcommand{\algorithmicgoto}{\textbf{go to}}%
\algnewcommand{\Goto}[1]{\algorithmicgoto~\ref{#1}}%

\usepackage[mathscr]{euscript}

\usepackage{graphicx}
\graphicspath{ {./} }
\DeclareGraphicsExtensions{.pdf}

\usepackage{subcaption}

\newtheorem{theorem}{Theorem}%
\newtheorem{lemma}{Lemma}%
\newtheorem{example}{Example}%
\newtheorem{remark}{Remark}%
\newtheorem{proof}{Proof}%

\usepackage{framed} %

\usepackage{geometry} %

\newenvironment{keywords}{
	\vspace{1em}\noindent\textbf{Keywords:}\ } %
{}

\begin{document}

\newgeometry{top=72pt,bottom=54pt,right=54pt,left=54pt}

\title{Fastest Mixing Reversible Markov Chain on Friendship Graph: Trade-Off between Transition Probabilities among Friends and Convergence Rate}
\author{Saber Jafarizadeh \\
	Rakuten Institute of Technology, Rakuten Crimson House, Tokyo, Japan \\
	\texttt{saber.jafarizadeh@rakuten.com}}

\date{}
\maketitle
\thispagestyle{empty} %

\bibliographystyle{plain}

\begin{abstract}

A long-standing goal of social network research has been to alter the properties of network to achieve the desired outcome.
In doing so, DeGroot's consensus model has served as the popular choice for modeling the information diffusion and opinion formation in social networks. 
Achieving a trade-off between the cost associated with modifications made to the network and the speed of convergence to the desired state has shown to be a critical factor. 
This has been treated as the Fastest Mixing Markov Chain (FMMC) problem over a graph with given transition probabilities over a subset of edges. 
Addressing this multi-objective optimization problem over the friendship graph, this paper has provided the corresponding Pareto optimal points or the Pareto frontier. 
In the case of friendship graph with at least three blades, it is shown that the Pareto frontier is reduced to a global minimum point which is same as the optimal point corresponding to the minimum spanning tree of the friendship graph, i.e., the star topology. 
	Furthermore, 
	a lower limit for transition probabilities among friends has been provided, 
	where values higher than this limit do not have any impact on the convergence rate.

\end{abstract}


\section{Introduction}
\label{sec:Introduction}

Social networking services have provided massive global infrastructure for individuals to exchange ideas and influence each other.
This phenomenon has been extensively studied 
for 
modeling the exchange of information 
and 
the impact of 
influence among individuals on processes such as opinion formation, influence  propagation, and information diffusion.
Majority of the developed models in the literature are based on the consensus models developed by 
DeGroot \cite{Degroot1974} and Lehrer \cite{Lehrer1975}, 
with %
additional features such as interaction frequency \cite{Patterson2010}.
Initial approaches are 
focused on structural properties of the graph model, including vertex degree, betweenness and distance centralities \cite{Wasserman1994}.
Subsequent developments take into account dynamic rather than static characteristics of the network's graph model \cite{Patterson2010,Viswanath2009,Wilson2009,Kossinets2008}.
A common goal of the developed models is to alter the properties of the network to achieve the desired outcome, such as encouraging the support for a cause or purchase of certain services or products.
In such scenarios, an important factor is to achieve a trade-off between the cost associated with alterations made to the network and the speed of convergence to the desired equilibrium state. 
This problem can be treated as the Fastest Mixing Markov Chain (FMMC) problem over a graph where the transition probabilities over a subset of edges are given.
Fastest Mixing Markov Chain (FMMC) problem is the problem of optimizing the mixing rate or the asymptotic convergence rate of Markov chain over a graph, by mimicking the transition probabilities.
In \cite{BoydFastestmixing2003}, it is shown that the asymptotic convergence rate of the Markov chain to the equilibrium distribution is determined by the Second Largest Eigenvalue Modulus ($SLEM$) of the transition probability matrix, %
which can be formulated as a Semidefinite Programming (SDP) problem. %
Within this context, %
authors in \cite{Gnecco2015,Gnecco2014} have addressed the problem through $l_{1}$-norm and $l_{0}$-"pseudo-norm" regularized versions of the FMMC problem, 
where they have provided theoretical and numerical results about several sparse variations of the problem. 
In \cite{Fardad2014}, 
the problem has been tackled by adding undirected edges to an existing graph, subject to a link-creation budget.
In this paper, we have considered the FMMC problem for reversible Markov chains (referred to as the FMRMC problem) with given transition probabilities over a subset of edges (denoted by fix edges), which can not be modified for the purpose of optimization. 
Here, 
we have addressed this multi-objective optimization problem by solving and formulating 
the %
corresponding Pareto frontier, which can be utilized for deriving a trade-off between the spectral properties of the topology (including the convergence rate) and the transition probabilities assigned to fix edges.
The topology considered in this paper is the friendship graph with the set of edges among friends as the fix edges. 
A friendship graph is a graph where every pair of vertices has exactly one common neighbor \cite{Erdos1966}.
Friendship graphs, due to their properties have been utilized in different disciplines, 
including 
block code design \cite{Leonard2005}, 
set theory \cite{LONGYEAR1972257} %
and 
graph labeling \cite{GallianAdynamic}.
In \cite{Bloch2018}, 
the authors have devised a mechanism for extracting ordinal information disseminated in a social network based on friend-based ranking, where the friendship graph has shown to be 
the sparsest topology for which the planner (i.e., the central vertex) can construct a complete ranking.
In the case of friendship graph with at least three blades (i.e., triangles), we have provided the optimal $SLEM$ 
and proved that it is same as that of the minimum spanning tree of the friendship graph, i.e., the star topology.
Furthermore, we have shown that 
the optimal $SLEM$ is independent of the transition probabilities over fix edges, 
and 
the Pareto frontier reduces to a single point, i.e., the minimum point. 
For the friendship graph with two blades, the Pareto Frontier is reduced to a minimum point only for equilibrium distribution satisfying a certain symmetry pattern.
For the friendship graph with one blade, the Pareto Frontier in always present. %
Friendship graph along with 
formulations of the reversible Markov chain over this graph and the FMRMC problem in terms of the symmetric Laplacian 
are introduced in Section \ref{sec:Preliminaries}. 
In Section \ref{sec:MainResults}, the main results are presented including 
the optimal weights and $SLEM$. %
The solution procedures are provided in Section \ref{sec:ProofofMainResults}. 
Section \ref{sec:Conclusions} concludes the paper.
\section{Preliminaries}
\label{sec:Preliminaries}
This section introduces the friendship graph followed by the formulations of the reversible Markov chain over this graph, and the FMRMC problem in terms of the symmetric Laplacian.

\subsection{Friendship Graph}
The friendship graph $\mathcal{F}_{m}$ is collection of $m$ triangles sharing a single common vertex, or in other words
a graph where every pair of its vertices has exactly one common neighbor.
This graph is depicted in Fig. \ref{fig:friendship}.
\begin{figure}
\centering
\begin{subfigure}[b]{0.234\hsize}
\includegraphics[width=\hsize]{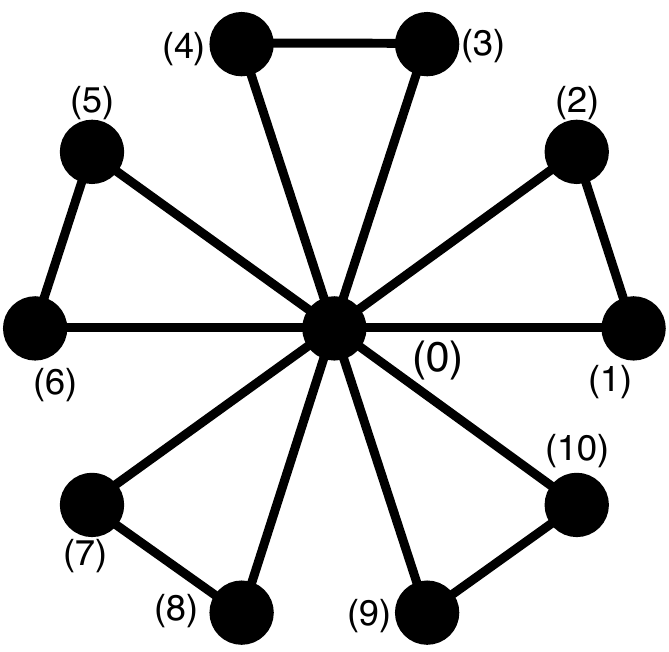}
\caption{}
\label{fig:friendship-m-5}
\end{subfigure}
\hspace{10pt}
\begin{subfigure}[b]{0.088\hsize}
\includegraphics[width=\hsize]{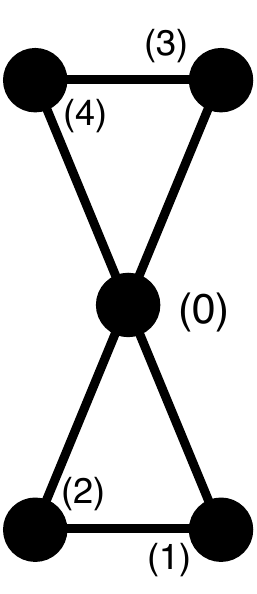}
\vspace{3pt}
\caption{}
\label{fig:friendship-m-2}
\end{subfigure}
\hspace{10pt}
\begin{subfigure}[b]{0.0845\hsize}
\includegraphics[width=\hsize]{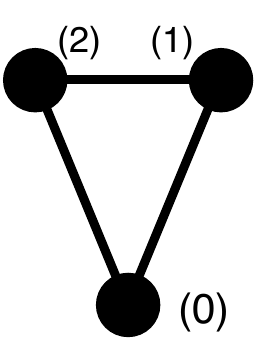}
\vspace{20pt}
\caption{}
\label{fig:friendship-m-1}
\end{subfigure}
\caption{Friendship graph with (a) $m=5$, (b) $m=2$ and (c) $m=1$ blades.}
\label{fig:friendship}
\end{figure}
The vertex set of this topology is
$\mathcal{V} = \{(0)\} \cup \{ (2i-1) | i=1,...,m \}\cup \{ (2i) | i=1,...,m \}$,
where $(0)$ denotes the central vertex of the star topology.
The edge set of friendship graph is 
$\mathcal{E}$ $=$ $\{ (0,2i-1) | i=1,...,m \}$ $\cup$ $\{ (0,2i) | i=1,...,m \}$ $\cup \{ (2i-1,$ $2i) | i=1,...,m \}$.
\subsection{Reversible Markov Chain on Friendship Graph}
\label{sec:RMCsection}

We consider an undirected friendship graph
$\mathcal{F}_{m}$    %
with vertex and edge sets $\mathcal{V}$ and $\mathcal{E}$, respectively.
Markov chain over the friendship graph can be defined by associating the state of Markov chain with vertices in the graph, i.e. $X(t) \in \mathcal{V}$, for $t=0,1,...$.
Correspondingly, transition of Markov chain between two adjacent states (vertices), is modeled by the edges in the graph, where each edge $\{i,j\}$ is associated with a transition probability $\boldsymbol{P}_{i,j}$.
For the Markov chain to be valid, the transition probability of each edge should be nonnegative (i.e. $\boldsymbol{P}_{i,j} \geq 0$) and the sum of the probabilities of outgoing edges of each vertex (including the self-loop edge) should be one.
A fundamental property of the irreducible Markov chains is the unique equilibrium distribution $(\pi)$, that satisfies
$\pi_{j}  =  \sum\nolimits_{i} \pi_{i} \boldsymbol{P}_{i,j}$ for all $j$ \cite{AldousBook2}.
The convergence theorem \cite{AldousBook2} states that for any initial distribution $P\left(X(t) = j \right) \rightarrow \pi_{j}$ as $t \rightarrow \infty$, for all $j$, provided the chain is aperiodic and irreducible.
The Markov chain is called reversible if the detailed balance equations holds true \cite{AldousBook2}, i.e.
\begin{equation}
\label{eq:Eq20171216220}
\begin{gathered}
\pi_{i} \boldsymbol{P}_{i,j}  =  \pi_{j} \boldsymbol{P}_{j,i} \;\; \text{for all} \;\; i,j.
\end{gathered}
\end{equation}
This paper is concerned with the rate of convergence of Markov chain to its equilibrium distribution.
This rate of convergence is determined by the eigenstructure of the transition probability matrix $\boldsymbol{P}$.
The eigenvalues of the transition probability matrix $\boldsymbol{P}$ are real, and no more than one in
magnitude \cite{AldousBook2}.
We denote %
the eigenvalues of the transition probability matrix $\boldsymbol{P}$ in non-increasing order, 
$-1 \leq \lambda_N \left( \boldsymbol{P} \right) \leq \lambda_{N-1} \left( \boldsymbol{P} \right) \leq \cdots \leq \lambda_1 \left( \boldsymbol{P} \right) = 1.$
Eigenstructure
of the transition matrix $\left( \boldsymbol{P} \right)$ %
determines the asymptotic convergence rate of the Markov chain to the equilibrium distribution $\left( \pi \right)$.
The asymptotic convergence rate of the Markov chain to the equilibrium distribution $\pi$ is determined by the Second Largest Eigenvalue Modulus (SLEM) of the transition probability matrix \cite{BoydFastestmixing2003}, defined as below
\begin{equation}
\label{eq:Eq20171216227}
\begin{gathered}
\mu \left( \boldsymbol{P} \right)
=
\max_{i=2, ..., N} \left| \lambda_{i} \left( \boldsymbol{P} \right) \right|
=
\max\{ \lambda_{2} \left( \boldsymbol{P} \right), -\lambda_{N} \left( \boldsymbol{P} \right) \}.
\end{gathered}
\end{equation}
Smaller SLEM results in faster convergence to the equilibrium distribution.
The quantity $\log\left( 1 / \mu \left( \boldsymbol{P} \right)  \right)$ is referred to as the mixing rate, and $\tau = 1 / \log\left( 1 / \mu \left( \boldsymbol{P} \right)  \right)$ as the mixing time \cite{BoydFastestmixing2003}.
The problem of optimizing the convergence rate of the reversible Markov chain to its equilibrium distribution is known as the Fastest Mixing Reversible Markov Chain (FMRMC) problem.
Let the transition probabilities among friends, i.e., the transition probabilities over the set of edges 
$\{ (2i-1,2i) | i=1,...,m \}$
be an additional costs which have to be minimized.
Using the detailed balance property (\ref{eq:Eq20171216220}), the transition probability from vertex $(i)$ to $(j)$, i.e., $\boldsymbol{P}_{i,j}$ can be determined from the transition probability in reverse direction, i.e., $\boldsymbol{P}_{j,i}$.
In %
this paper, we address the problem of optimizing the convergence rate of the reversible Markov chain with the additional objective of minimizing the transition probabilities among friends.
This can be formulated as the following
multi-objective optimization problem \cite{BoydConvexBook},
\begin{equation}
\label{eq:Eq20181027427}
\begin{aligned}
\min\limits_{ \{\boldsymbol{P}_{i,j} | \{i,j\} \in \mathcal{E} \} }
\;\;
&
	\boldsymbol{f}_{0} ( \boldsymbol{P} ) =	\left[	\mu\left(\boldsymbol{P}\right),	\boldsymbol{P}_{1,2},	...,	\boldsymbol{P}_{2m-1,2m},	\right]^{T}
\\
s.t.
\qquad\;\;
&
\boldsymbol{P} \geq 0,
\;\;
\boldsymbol{P} \boldsymbol{1} = \boldsymbol{1},
\;\;
\boldsymbol{\pi}^{T} \boldsymbol{P} = \boldsymbol{\pi}^{T},
\\
&
\boldsymbol{D} \boldsymbol{P} = \boldsymbol{P}^{T} \boldsymbol{D},
\;\;
\boldsymbol{P}_{i,j} = 0 \; \text{for} \; \{i,j\} \notin \mathcal{E}.
\end{aligned}
\end{equation}
where
$\boldsymbol{P} \geq 0$ is element-wise,
$\boldsymbol{D} = diag\left( \pi_{0}, \pi_{1}, ..., \pi_{2m} \right)$,
$\boldsymbol{1}$ is the vector of all one
and
$\boldsymbol{\pi} = [ \pi_{0}, \pi_{1}, \pi_{2}, ..., \pi_{2m} ]^{T}$.
The constraint $\boldsymbol{D} \boldsymbol{P} = \boldsymbol{P}^{T} \boldsymbol{D}$ imposes the detailed balance equation (\ref{eq:Eq20171216220}).
For two points in the feasible region $\boldsymbol{P} (1)$ and $\boldsymbol{P} (2)$, the inequality $\boldsymbol{f}_{0} ( \boldsymbol{P} (1) )  \leq  \boldsymbol{f}_{0} ( \boldsymbol{P} (2) )$ between their corresponding objective functions is element-wise, i.e.,
$\mu( \boldsymbol{P} (1) )$  $\leq$  $\mu( \boldsymbol{P} (2) )$,
$\boldsymbol{P}_{1,2} ( \boldsymbol{P} (1) )$   $\leq$  $\boldsymbol{P}_{1,2} ( \boldsymbol{P} (2) )$,
$...$,
$\boldsymbol{P}_{2m-1,2m} ( \boldsymbol{P} (1) )$   $\leq$  $\boldsymbol{P}_{2m-1,2m} ( \boldsymbol{P} (2) )$.
Considering $\mathcal{S}$ as the set of feasible points,
point $\boldsymbol{P}(1)$ is a minimal element of $\mathcal{S}$ if $\boldsymbol{P}(2) \in \mathcal{S}$ and  $\boldsymbol{P}(2) \leq \boldsymbol{P}(1)$ only if $\boldsymbol{P}(2) = \boldsymbol{P}(1)$,
and
point $\boldsymbol{P}(1)$ is a minimum element of $\mathcal{S}$ if for every  $\boldsymbol{P}(2) \in \mathcal{S}$, we have $\boldsymbol{P}(1) \leq \boldsymbol{P}(2)$.
The notion of minimal element is weaker than that of the minimum element. 
A minimal element of $\mathcal{S}$ is also denoted as a Pareto optimal point.
The set of Pareto optimal points
is called the Pareto frontier.
To determine the Pareto frontier of the multi-objective optimization problem (\ref{eq:Eq20181027427}),
we minimize the $SLEM$ for given values of the transition probabilities over the subset of edges between friends (i.e., 
$\mathcal{E}_{F} = \{ (2i-1,2i) | i=1,...,m \}$).
The resultant optimal points, determine a part of boundary of the feasible region
that contains the Pareto frontier.
Hence, problem (\ref{eq:Eq20181027427}) can be interpreted as the following  optimization problem,
\begin{equation}
\label{eq:Eq201712161073}
\begin{aligned}
\hspace{-10pt}
\min\limits_{ \{\boldsymbol{P}_{i,j} | \{i,j\} \in \mathcal{E}_{V} \} }
\;\;
&\mu\left(\boldsymbol{P}\right),
\\
s.t.
\qquad\;\;
&
\boldsymbol{P} \geq 0,
\;\;
\boldsymbol{P} \boldsymbol{1} = \boldsymbol{1},
\;\;
\boldsymbol{\pi}^{T} \boldsymbol{P} = \boldsymbol{\pi}^{T},
\\
&
\boldsymbol{D} \boldsymbol{P} = \boldsymbol{P}^{T} \boldsymbol{D},
\;\;
\boldsymbol{P}_{i,j} = 0 \; \text{for} \; \{i,j\} \notin \mathcal{E}.
\end{aligned}
\end{equation}
where
$\mathcal{E}_{V} = \{ (0,2i-1) \} \cup \{ (0,2i) \}$ for $i=1,...,m$.
Note that $\mathcal{E}_{F}$ and $\mathcal{E}_{V}$ are disjoint subsets and $\mathcal{E}_{F} \cup \mathcal{E}_{V} = \mathcal{E}$.
\subsection{FMRMC in terms of Symmetric Laplacian}
\label{sec:FMRMCSymmetricLaplacian}
In (\ref{eq:Eq201712161073}), the constraint $\boldsymbol{D} \boldsymbol{P} = \boldsymbol{P}^{T} \boldsymbol{D}$ adds certain difficulties to the solution of (\ref{eq:Eq201712161073}). %
To overcome this problem,
here, %
we have converted the optimization problem (\ref{eq:Eq201712161073}) into an optimization problem with symmetric optimization variables.
Defining the asymmetric Laplacian as below,
\begin{equation}
\nonumber
\begin{aligned}
\boldsymbol{L} \left( \boldsymbol{P} \right)_{ i,j }
=
\begin{cases}
- \boldsymbol{P}_{i,j} &\text{if} \quad i \neq j,
\\
\sum_{j} \boldsymbol{P}_{i,j} &\text{if} \quad i = j,
\end{cases}
\end{aligned}
\end{equation}
the transition probability matrix $\left( \boldsymbol{P} \right)$ in terms of the asymmetric Laplacian matrix
$\left( \boldsymbol{L} \left( \boldsymbol{P} \right) \right)$
can be written as
$\boldsymbol{P} = \boldsymbol{I} - \boldsymbol{L} \left( \boldsymbol{P} \right)$.
Note that the only constraint on $\pi_{i}$ is that $\pi_{i} > 0$ for $i\in\mathcal{V}$, and there is no constraint for their summation to be equal to one.
Defining the weights $q_{i,j}$ as
\begin{equation}
\label{eq:Eq20171223512}
\begin{gathered}
q_{i,j} = \pi_{i} \boldsymbol{P}_{i,j},
\quad \text{for} \quad \{i,j\} \in \mathcal{E}
\end{gathered}
\end{equation}
and imposing the detailed balance condition (\ref{eq:Eq20171216220})
(which results in symmetric $q_{i,j}$, i.e. $q_{i,j} = \pi_{i} \boldsymbol{P}_{i,j} = \pi_{j} \boldsymbol{P}_{j,i} = q_{j,i}$),
for the transition probability matrix 
$\boldsymbol{P}$,    %
we have
\begin{equation}
\label{eq:Eq20171112287}
\begin{gathered}
\boldsymbol{P}
=
\boldsymbol{I}  -  \boldsymbol{L} \left( \boldsymbol{P} \right)
=
\boldsymbol{D}^{-1} \left( \boldsymbol{D} - \boldsymbol{L} \left( q \right)  \right)
=
\boldsymbol{I} - \boldsymbol{D}^{-1}  \boldsymbol{L} \left( q \right)
\end{gathered}
\end{equation}
where $\boldsymbol{L}(q)$ is the symmetric Laplacian matrix defined as
\begin{equation}
\label{eq:Eq20171112316}
\begin{gathered}
\boldsymbol{L}(q) = \sum\nolimits_{ \{i,j\} \in \mathcal{E} } q_{ij} \left( \boldsymbol{e}_{i} - \boldsymbol{e}_{j} \right) \left( \boldsymbol{e}_{i} - \boldsymbol{e}_{j} \right)^{T}.
\end{gathered}
\end{equation}
$\boldsymbol{e}_{i}$ for $i \in \mathcal{V}$ is a column vector with $i$-th element equal to $1$ and zero elsewhere.
The right eigenvector of
the transition probability matrix
$\boldsymbol{P}$
corresponding to eigenvalue one is $\boldsymbol{1}$
(the vector of all ones),
i.e.
$\boldsymbol{P} \boldsymbol{1} = \left( \boldsymbol{I} - \boldsymbol{D}^{-1} \boldsymbol{L}(q) \right) \boldsymbol{1} = \boldsymbol{1}$. 
The left eigenvector of
the transition probability matrix
$\boldsymbol{P}$
corresponding to eigenvalue one is
$\boldsymbol{\pi}^{T} = \left[ \pi_{0}, \pi_{1}, ..., \pi_{2m} \right]$, i.e.
\begin{equation}
\label{eq:Eq20171112356}
\begin{gathered}
\boldsymbol{\pi}^{T} \boldsymbol{P}
=
\boldsymbol{\pi}^{T} \left( \boldsymbol{I} - \boldsymbol{D}^{-1} \boldsymbol{L}(q) \right)
=
\boldsymbol{\pi}^{T}
\end{gathered}
\end{equation}
Even though the transition probability matrix
$\boldsymbol{P}$
is not symmetric but it has real spectrum \cite{AldousBook2}.
For eigenvalues ($\lambda_n$) and eigenvectors ($\boldsymbol{\phi}_n$) of 
$\boldsymbol{P}$,
we have
$  \boldsymbol{P} \boldsymbol{\phi}_n    =    ( \boldsymbol{I} - \boldsymbol{D}^{-1} \boldsymbol{L}(q) ) \boldsymbol{\phi}_n    =$    $\lambda_n \boldsymbol{\phi}_n
$
for $n=1,...,N$.
Considering 
orthonormality of the eigenvectors $\boldsymbol{\phi}_{n}$ and
$\boldsymbol{P} = \sum\nolimits_{n=1}^{N}  \lambda_n \boldsymbol{\phi}_{n} \boldsymbol{\phi}_{n}^{T} \boldsymbol{D}$, 
$ \boldsymbol{\phi}_{1} \boldsymbol{\phi}_{1}^{T} = \frac{ \boldsymbol{J} }{ \sum_{i} \pi_{i} }$, 
$ \lambda_1 = 1 $, 
it can be shown that 
$\boldsymbol{P}^{k}$ for $k \rightarrow \infty$ converges to the unique answer $ \frac{ \boldsymbol{J} \boldsymbol{D} }{ \sum_{i} \pi_{i} } $, 
or equivalently $\boldsymbol{\pi}$ is the unique equilibrium distribution corresponding to transition probability matrix $\boldsymbol{P}$
\cite{AldousBook2}, 
if 
all eigenvalues $\lambda_{n}$ for $n=2$,...,$N$ are smaller than $1$ in absolute value, or in other words 
the spectral radius of $\boldsymbol{P}$ $-$ $\frac{\boldsymbol{J}\boldsymbol{D}}{\sum_{i}\pi_{i}}$ is not greater than one, i.e., 
$\rho ( \boldsymbol{P} - \frac{\boldsymbol{J}\boldsymbol{D}}{\sum_{i}\pi_{i}} )  =  \rho ( \boldsymbol{I} - \boldsymbol{D}^{-1}$ $\boldsymbol{L}(q)$ $-$ $\frac{\boldsymbol{J}\boldsymbol{D}}{\sum_{i}\pi_{i}} ) \leq 1$, 
which in turn is equivalent to following, 
\begin{equation}
\label{eq:Eq20171115740}
\begin{gathered}
-\boldsymbol{I}
<
\boldsymbol{D}^{\frac{1}{2}}
\left( \boldsymbol{P}  -  \frac{\boldsymbol{J}\boldsymbol{D}}{\sum_{i}\pi_{i}} \right)
\boldsymbol{D}^{-\frac{1}{2}}
=
\left( \boldsymbol{I} - \boldsymbol{D}^{-\frac{1}{2}} \boldsymbol{L}\left( q \right) \boldsymbol{D}^{-\frac{1}{2}} \right)  -  \widetilde{\boldsymbol{J}}
<
\boldsymbol{I}
\end{gathered}
\end{equation}
where
$\widetilde{\boldsymbol{J}}  =  \frac{ \boldsymbol{D}^{\frac{1}{2}} \boldsymbol{J} \boldsymbol{D}^{\frac{1}{2}}  }{ \sum_{i} \pi_{i} }$
with
$\boldsymbol{J}  =  \boldsymbol{1} \boldsymbol{1}^{T}$.
The FMRMC problem (\ref{eq:Eq201712161073}) in terms of the symmetric Laplacian can be interpreted as the following optimization problem,
\begin{equation}
\label{eq:Eq201712231257} %
\begin{aligned}
\hspace{-10pt}
\min\limits_{ \{ q_{i,j} | \{i,j\} \in \mathcal{E}_{V} \} }
\;\;
&\mu\left( \boldsymbol{I} - \boldsymbol{D}^{-1} \boldsymbol{L}\left( q \right) \right),
\\
s.t.
\qquad\;\;
&
q_{i,j} = q_{j,i} \geq 0,
\;\;
q_{i,j} = 0 \; \text{for} \; \{i,j\} \notin \mathcal{E},
\\ &
\boldsymbol{\pi}^{T} \left( \boldsymbol{I} - \boldsymbol{D}^{-1} \boldsymbol{L}\left( q \right) \right)  =  \boldsymbol{\pi}^{T},
\;\; %
\left( \boldsymbol{I} - \boldsymbol{D}^{-1} \boldsymbol{L}\left( q \right) \right) \boldsymbol{1}
=
\boldsymbol{1},
\\
&
- \sum\nolimits_{k \neq i} q_{i,k} + \pi_{i} \geq 0 \;\; \text{for} \;\;
i \in \mathcal{V}  %
\end{aligned}
\end{equation}
\section{Main Results}
\label{sec:MainResults}
In this section, we present the main results of the paper regarding the FMRMC problem as formulated in (\ref{eq:Eq201712231257}) over a friendship graph.
These results include the weights and $SLEM$ corresponding to the Pareto frontier and the minimum points (if it exists) for given equilibrium distribution and weights over edges between friends $\left( \mathcal{E}_{F} \right)$.
Proofs and detailed solutions are deferred to Section \ref{sec:ProofofMainResults}.

For
$m \geq 3$ and
$\Pi \leq 2\pi_{0}$,
with 
$\Pi = \sum_{i=1}^{m} \left( \pi_{2i-1} + \pi_{2i} \right)$
and
\begin{equation}
\label{eq:corner-constraints-star-interior}
\begin{gathered}
0
\leq
q_{2i-1,2i}
\leq
\left(  \Pi  / \left(  2 \pi_{0} + \Pi \right) \right) 
\min\{ \pi_{2i-1} , \pi_{2i} \}
\end{gathered}
\end{equation}
the optimal weight and the optimal $SLEM$ are as below,
\begin{equation}
\label{eq:Eq201807237981}
\begin{gathered}
q_{0,2i-1} %
=   2 \pi_{0} \pi_{2i-1} / \left( 2 \pi_{0}  + \Pi \right),
\;\;\;\; %
q_{0,2i} %
=  2 \pi_{0} \pi_{2i} / \left( 2 \pi_{0}  + \Pi \right),
\end{gathered}
\end{equation}
\begin{equation}
\label{eq:Eq201807237991}
\begin{gathered}
SLEM %
=   \Pi   /  \left(    2 \pi_{0}  + \Pi    \right),
\end{gathered}
\end{equation}
For
$m \geq 3$ and
$\Pi \geq 2\pi_{0}$,
and
\begin{equation}
\label{eq:corner-constraints-star-non-interior}
\begin{gathered}
0
\leq
q_{2i-1,2i} %
\leq
\left( \left(  \Pi - \pi_{0}  \right)  /  \Pi  \right) 
\min\{ \pi_{2i-1}, \pi_{2i} \} %
\end{gathered}
\end{equation}
the optimal weight and the optimal $SLEM$ are as below,
\begin{equation}
\label{eq:Eq201807278235}
\begin{gathered}
q_{0,2i-1} %
=
\pi_{2i-1} \pi_{0}  /  \Pi,    %
\;\;\;\;\; %
q_{0,2i}    %
=
\pi_{2i} \pi_{0} / \Pi    %
\end{gathered}
\end{equation}
\begin{equation}
\label{eq:Eq201807278283}
\begin{gathered}
SLEM %
=
\left(  \Pi - \pi_{0}  \right)  /  \Pi  ,
\end{gathered}
\end{equation}
where (\ref{eq:Eq201807237981}) and (\ref{eq:Eq201807278235}) holds for $i=1,...,m$.

\begin{lemma} \label{lemma:star-high-q-given}
Values of 
$q_{2i-1,2i}$    %
greater than the limits provided in
(\ref{eq:corner-constraints-star-interior}) and (\ref{eq:corner-constraints-star-non-interior}) %
result in optimal value of $SLEM$
greater than the value of $SLEM$ obtained for
values of 
$q_{2i-1,2i}$    %
that satisfy
(\ref{eq:corner-constraints-star-interior}) and (\ref{eq:corner-constraints-star-non-interior}). %
\end{lemma}	
The proof of this lemma is provided in Appendix \ref{lemma:star-high-q-given-proof}.
\begin{remark}
From Lemma \ref{lemma:star-high-q-given}, it can be concluded that the optimal values of $SLEM$ in (\ref{eq:Eq201807237991}) and (\ref{eq:Eq201807278283}) are the smallest values of $SLEM$ for all values of 
$q_{2i-1,2i}$    %
for $i=1,...,m$.
\end{remark}
\begin{theorem}
\label{theorem1209}
For friendship topologies with $m \geq 3$, the Pareto frontier of problem (\ref{eq:Eq20181027427}) is a single point corresponding to 
$q_{2i-1,2i} = 0$ for $i=1,...,m$, %
i.e., this point is the minimum point of the optimization problem (\ref{eq:Eq20181027427}).
\end{theorem}
\begin{proof}
From the results represented in (\ref{eq:Eq201807237981}), (\ref{eq:Eq201807237991}),  (\ref{eq:Eq201807278235}) and (\ref{eq:Eq201807278283}) for the optimization problem (\ref{eq:Eq201712161073}), it can be concluded that for $m \geq 3$ the Pareto frontier of the optimization problem (\ref{eq:Eq20181027427}) is a single point corresponding to 
$q_{2i-1,2i} = 0$    %
for $i=1,...,m$.
Thus, this point is the minimum point of the 
problem (\ref{eq:Eq20181027427}).
\end{proof}
\begin{remark}
\label{remark:star}
The optimal results presented in (\ref{eq:Eq201807237981}), (\ref{eq:Eq201807237991}),  (\ref{eq:Eq201807278235}) and (\ref{eq:Eq201807278283}) are same as those of the star topology which represents the minimum spanning tree of the friendship graph.
\end{remark}
From remark \ref{remark:star}, it is obvious that the weights or the transition probabilities over the edges
between friends, i.e., $\mathcal{E}_{F}$ $=$ $\{ (2i-1$, $2i) |$ $i=1,...,m \}$
have no impact on the optimal $SLEM$ (i.e., the convergence rate of the reversible Markov chain problem to its equilibrium distribution),
as long as they are smaller than the upper limits provided in (\ref{eq:corner-constraints-star-interior})
and
(\ref{eq:corner-constraints-star-non-interior}).

\begin{remark}
\label{remark:201902032012}
From    %
(\ref{eq:Eq201807237991}) 
and    %
(\ref{eq:Eq201807278283}), 
it is obvious that for $\Pi \leq 2 \pi_{0}$ (i.e., the equilibrium distribution of central vertex is more than half of the rest of vertices), the convergence rate is faster and it is smaller than $\frac{1}{2}$.
While, for the case of $\Pi \geq 2 \pi_{0}$, it is opposite, i.e.,  the optimal convergence rate is slower and it is larger than $\frac{1}{2}$.
Another interesting point is regarding the optimal transition probabilities of the edges connected to the central vertex.
From
(\ref{eq:Eq201807237981}) 
and    %
(\ref{eq:Eq201807278235}), 
it is obvious that the optimal transition probabilities on the edges heading towards the central vertex (i.e., $\boldsymbol{P}_{i,0}$ for $i=1,...,2m$) are all equal to each other.
While the same does not hold true for the edges leaving the central vertex (i.e., $\boldsymbol{P}_{0,i}$ for $i=1,...,2m$).
\end{remark}
\subsection{Case of $m=2$}
For
$m=2$,
and equilibrium distributions that satisfy
\begin{equation}
\label{eq:Eq201810142587}
\begin{gathered}
\pi_{0}^{2}    \geq    \left( \pi_{1} + \pi_{2} \right)    \left( \pi_{3} + \pi_{4} \right)
\end{gathered}
\end{equation}
the optimal weights and the optimal $SLEM$ are as below,
\begin{equation}
\label{eq:Eq201810132536}
\begin{gathered}
q_{0,2i+j-2)} = 
\left( \pi_{0} \pi_{2i+j-2} \right) / \left( \pi_{0} + \pi_{2i-1} + \pi_{2i} \right)
\; \text{for} \; i,j=1,2,
\end{gathered}
\end{equation}
\begin{equation}
\label{eq:Eq201810132576}
\begin{gathered}
SLEM = 
\sqrt{ \frac{  \left( \pi_{1} + \pi_{2} \right) \left( \pi_{3} + \pi_{4} \right)  }{  \left( \pi_{0} + \pi_{1} + \pi_{2} \right) \left( \pi_{0} + \pi_{3} + \pi_{4} \right)  } }
\end{gathered}
\end{equation}
if
and $q_{1,2}$ and $q_{3,4}$ are within 
following boundaries
\begin{subequations}
\label{eq:201810142672}
\begin{gather}
A_{1}    %
-  s
<
q_{1,2}    %
<
A_{1}     %
\min\{ \pi_{1}, \pi_{2} \},    %
\label{eq:201810142672a}
\\
A_{2}    %
-  s
<
q_{3,4}    %
<
A_{2}    %
\min\{ \pi_{3}, \pi_{4} \}.    %
\label{eq:201810142672b}
\end{gather}
\end{subequations}
where
$A_{1} = \frac{  \pi_{1} + \pi_{2}  }{  \pi_{0} + \pi_{1} + \pi_{2}  }$,
$A_{2} = \frac{  \pi_{3} + \pi_{4}  }{  \pi_{0} + \pi_{3} + \pi_{4}  }$,
and
$s$ is the optimal $SLEM$ as provided in (\ref{eq:Eq201810132576}).
Note that for 
$\pi_{1} + \pi_{2} > \pi_{3} + \pi_{4}$, 
the lower bound in
(\ref{eq:201810142672b}) is negative and only the lower bound in (\ref{eq:201810142672a}) imposes a lower limit on 
$q_{1,2}$.    %
Similarly for 
$\pi_{1}$ $+$ $\pi_{2}$ $<$ $\pi_{3}$ $+$ $\pi_{4}$, 
the lower bound in
(\ref{eq:201810142672a}) is negative and only the lower bound in (\ref{eq:201810142672b}) imposes a lower limit on 
$q_{3,4}$.    %
For the equilibrium distribution where 
$\pi_{1} + \pi_{2} = \pi_{3} + \pi_{4}$, 
both the lower bounds in (\ref{eq:201810142672}) are equal to zero.

For equilibrium distributions that does not satisfy (\ref{eq:Eq201810142587}),
the optimal weights and the optimal $SLEM$ are as below,
\begin{equation}
\label{eq:Eq201810162962}
\begin{gathered}
q_{0,2i+j-2} =
\pi_{2i+j-2} \pi_0 ( \pi_0 + 2 ( \pi_{2i^{'}-1} + \pi_{2i^{'}} ) )  /  A_{0} 
\end{gathered}
\end{equation}
\begin{equation}
\label{eq:Eq201810142678}
\begin{gathered}
SLEM %
=  
\left(  4 \left( \pi_{1} + \pi_{2} \right) \left( \pi_{3} + \pi_{4} \right) - \pi_{0}^{2}  \right)  /  A_{0}
\end{gathered}
\end{equation}
where
$A_{0} = \pi_0 ( \pi_{1} + \pi_{2}  +  \pi_{3} + \pi_{4}  )  +  4 ( \pi_{1} + \pi_{2} ) ( \pi_{3} + \pi_{4} )$
and 
in (\ref{eq:Eq201810162962}), $i,j=1,2$ and $i^{'} = (i + 1 ) \; mod \; 2$. %
The results in (\ref{eq:Eq201810162962}) and  (\ref{eq:Eq201810142678}) hold true,
if
and 
$q_{1,2}$ and $q_{3,4}$
are within the following boundaries
\begin{subequations}
\label{eq:201810163031}    %
\begin{gather}
q_{1,2}     \geq		\frac		{  \pi_{0} \left( \left( \pi_{1} + \pi_{2} \right) - \left( \pi_{3} + \pi_{4} \right) \right) }{  A_{0}  }
\label{eq:201810163031a}
\\
q_{1,2}   \leq		\left( 1 -  \frac{  \pi_0 \left( \pi_0 + 2 \left( \pi_{3} + \pi_{4} \right) \right)  }{  A_{0}  } \right)  \min\{ \pi_{1}, \pi_{2} \},
\label{eq:201810163031b}
\\
q_{3,4}  \geq	\frac	{  \pi_{0} \left( \left( \pi_{3} + \pi_{4} \right) - \left( \pi_{1} + \pi_{2} \right) \right)  }{  A_{0}  }
\label{eq:201810163031c}
\\
q_{3,4}		\leq		\left( 1 - \frac{  \pi_0 \left( \pi_0 + 2 \left( \pi_{1} + \pi_{2} \right) \right)   }{  A_{0}  } \right)	\min\{ \pi_{3}, \pi_{4} \},
\label{eq:201810163031d}
\end{gather}
\end{subequations}
From (\ref{eq:201810163031a}) and (\ref{eq:201810163031c}), it is obvious that
for 
$\pi_{1} + \pi_{2} > \pi_{3} + \pi_{4}$, 
the bound in
(\ref{eq:201810163031c}) is negative and only (\ref{eq:201810163031a}) imposes a lower limit on 
$q_{1,2}$.    %
Similarly for 
$\pi_{1} + \pi_{2} < \pi_{3} + \pi_{4}$, 
the bound in
(\ref{eq:201810163031a}) is negative and only (\ref{eq:201810142672b}) imposes a lower limit on 
$q_{3,4}$.    %
\begin{remark}
\label{remark:m2-equal-sides}
For the case of equilibrium distributions with 
$\pi_{1} + \pi_{2} = \pi_{3} + \pi_{4}$, 
the lower bounds in 
(\ref{eq:201810142672}) 
and 
(\ref{eq:201810163031})  %
are equal to zero and therefore, the Pareto frontier reduces to a single point which in turn is the minimum point of problem (\ref{eq:Eq20181027427}).
This is true for both type of equilibrium distributions that satisfy (\ref{eq:Eq201810142587}) and those that not satisfy (\ref{eq:Eq201810142587}).
\end{remark}
Similar to the case of $m \geq 3$ (Lemma \ref{lemma:star-high-q-given}), for $m=2$ it can be shown that
values of 
$q_{2j+1,2j}$    %
for $j=1,2$ greater than the upper limits provided in (\ref{eq:201810142672}) and (\ref{eq:201810163031}) result larger optimal $SLEM$ compared to those obtained for values of 
$q_{2j-1,2j}$    %
for $j=1,2$ within the boundaries given in (\ref{eq:201810142672}) and (\ref{eq:201810163031}).
In general, providing the closed-form formula for the Pareto frontier in the case of equilibrium distributions with
$\pi_{1}$ $+$ $\pi_{2}$ $\neq$ $\pi_{3}$ $+$ $\pi_{4}$
can be cumbersome. %
In the following, we provide two examples for friendship graph with $m=2$, where we have provided the Pareto frontier.
\begin{example}
\label{example:1}
Consider the friendship graph with $m=2$ and the equilibrium distribution
$\pi_{0} = \pi_{4} = 3$, $\pi_{1}$ $=$ $\pi_{2} =$ $\pi_{3} = 1$.
This equilibrium distribution satisfies (\ref{eq:Eq201810142587}).
The lower bound in (\ref{eq:201810142672b}) imposes the lower limit
$\frac{ 3 }{ 7 } - \frac{ 3 }{ \sqrt{70} }$
on 
$q_{3,4}$    %
and
the optimal values of $s$ for
$0 \leq q_{3,4} \leq \frac{ 3 }{ 7 } - \frac{ 3 }{ \sqrt{70} }$ 
is 
$s = (  ( 3 - 7  q_{3,4}  )  /  4  )  - ( 1 / 6.9282 ) \sqrt{  35 q_{3,4}^{2}  -  30 q_{3,4}  +  3  } $.
This curve (as depicted in Fig \ref{fig:ParetoExample_1}) is the Pareto Frontier for this example.
\end{example}
\begin{figure}
\centering
\begin{subfigure}[b]{0.325\hsize}
\includegraphics[width=\hsize]{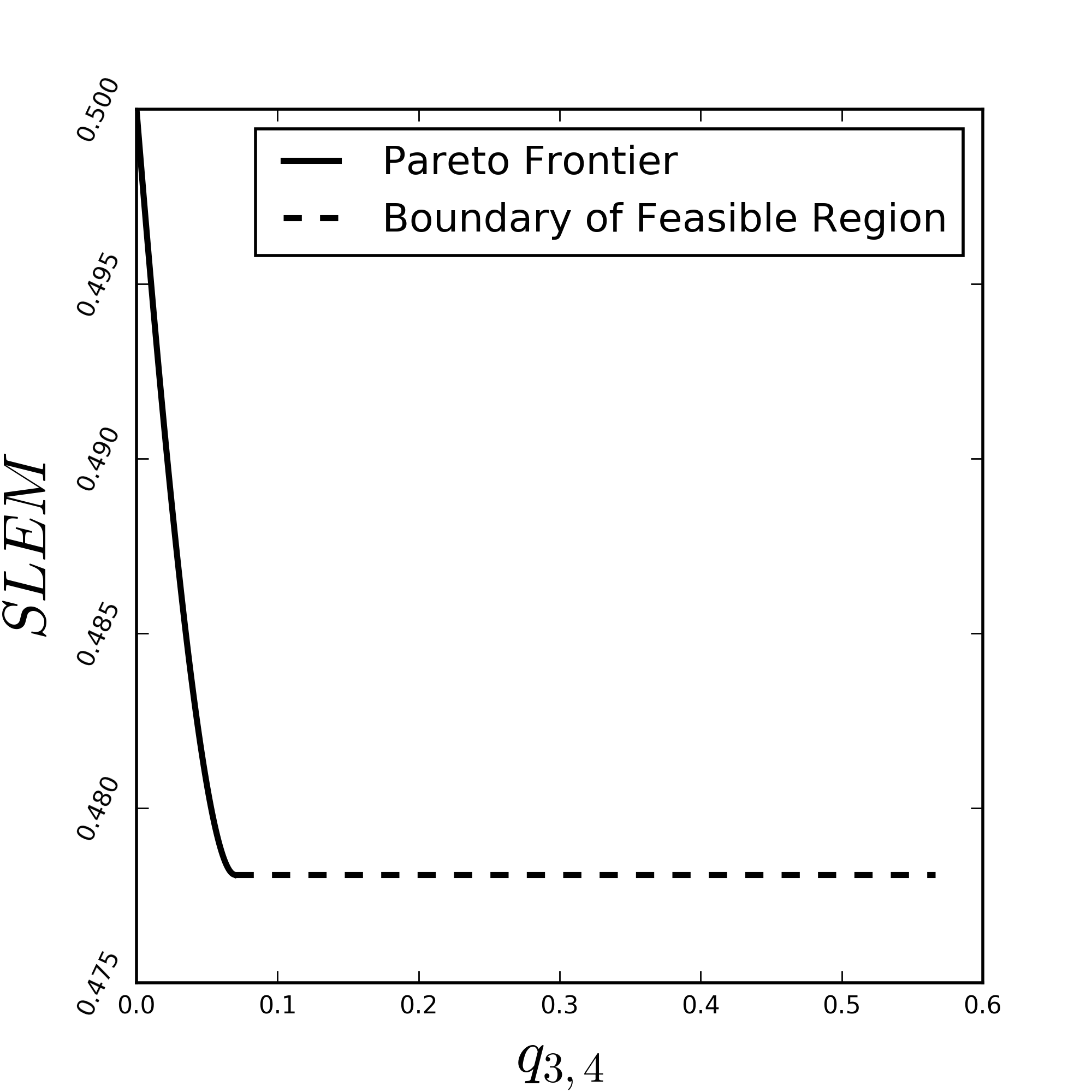}
\caption{}
\label{fig:ParetoExample_1}
\end{subfigure}
\hspace{-14pt}
\begin{subfigure}[b]{0.325\hsize}
\includegraphics[width=\hsize]{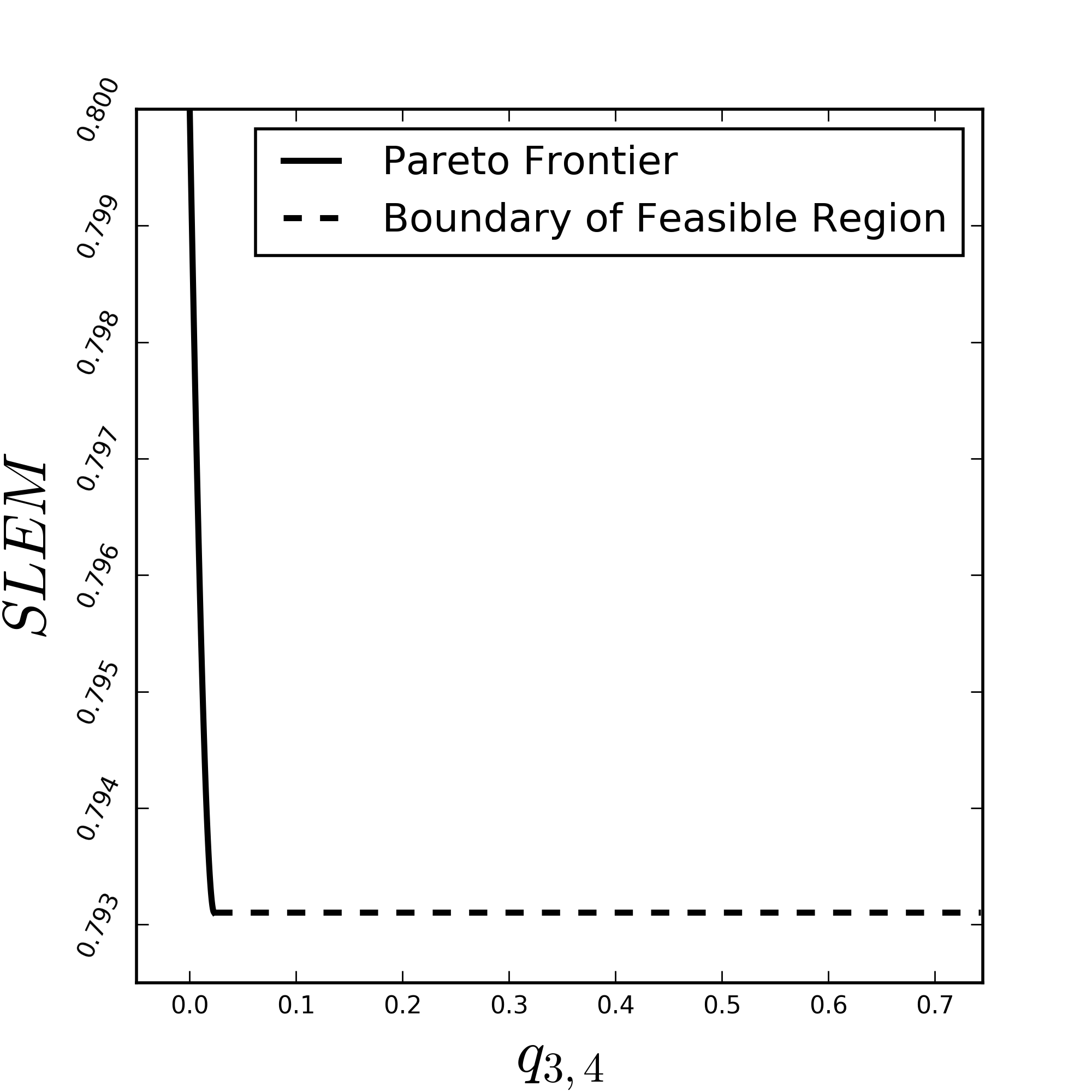}
\caption{}
\label{fig:ParetoExample_2}
\end{subfigure}
\caption{Pareto frontier of (a) Example \ref{example:1} and (b) Example \ref{example:2}.}
\end{figure}

\begin{example}
\label{example:2}
Consider the friendship graph with $m=2$ and the equilibrium distribution
$\pi_{4} = 2$, $\pi_{0} = \pi_{1} = \pi_{2} = \pi_{3} = 1$.
This equilibrium distribution does not satisfy (\ref{eq:Eq201810142587}).
The lower bound in (\ref{eq:201810163031c}) imposes the lower limit
$\frac{ 2 }{ 87 }$
on 
$q_{3,4}$    %
and
the optimal values of $s$ for
$0 \leq q_{3,4} \leq \frac{ 2 }{ 87 }$
is as below,
$s = (  18$ $-$ $35 q_{3,4} - \sqrt{ 145 q_{3,4}^{2} - 100 q_{3,4} + 4 }  ) /  20  $.
This curve (as depicted in Fig \ref{fig:ParetoExample_2}) is the Pareto Frontier for this example.
\end{example}
Note that in both examples above, the optimal value of $s$ is independent of 
$q_{1,2}$.    %

\subsection{Case of $m=1$}
\label{sec:main-results-m1}
In the case of friendship graph with $m=1$, the graph is reduced to a triangle.
We denote the central vertex by index $3$ and other two vertices by indices $1$ and $2$, where $q_{1,2}$ is given.
\subsubsection{Equilibrium distributions with $\pi_{3}^{2} > \pi_{1} \pi_{2}$}
For this case,
independent of the values of $q_{1,2}$,
the optimal weights and $SLEM$ are as below,
\begin{equation}
\label{eq:201809151209-main-result}
\begin{gathered}
q_{1,3} = 
\pi_{3} ( \pi_{1} - q_{1,2} )  / \left(  \pi_{1} + \pi_{3}  \right),
\;\; %
q_{2,3} = 
\pi_{3} ( \pi_{2} - q_{1,2} )  /  \left(  \pi_{2} + \pi_{3}  \right),
\end{gathered}
\end{equation}
the optimal weights and the optimal $SLEM$ are as below,
\begin{equation}
\label{eq:201809151244-main-result}
\begin{gathered}
s = 
\left(  \left| \pi_{1} \pi_{2} - \left( \pi_{1} + \pi_{2} + \pi_{3} \right) q_{1,2} \right|  \right)    /
\sqrt{ \pi_{1} \pi_{2} \left( \pi_{1} + \pi_{3} \right)  \left( \pi_{2} + \pi_{3} \right) }
\end{gathered}
\end{equation}
It is obvious that the Pareto frontier for
this case
is the line
$ s
=
(  \pi_{1} \pi_{2} - ( \pi_{1} + \pi_{2} + \pi_{3} ) q_{1,2}  )     /
\sqrt{ \pi_{1} \pi_{2} \left( \pi_{1} + \pi_{3} \right)  \left( \pi_{2} + \pi_{3} \right) } $
for
$  0  \leq  q_{1,2}  \leq  \frac{  \pi_{1} \pi_{2}  }{  \pi_{1} + \pi_{2} + \pi_{3}  }  $,
which is depicted in Fig. \ref{fig:Friendship-m1-p3geqp1p2-Pareto-Frontier}.
\begin{figure}
\centering
\begin{subfigure}[b]{0.4\hsize}
\includegraphics[width=\hsize]{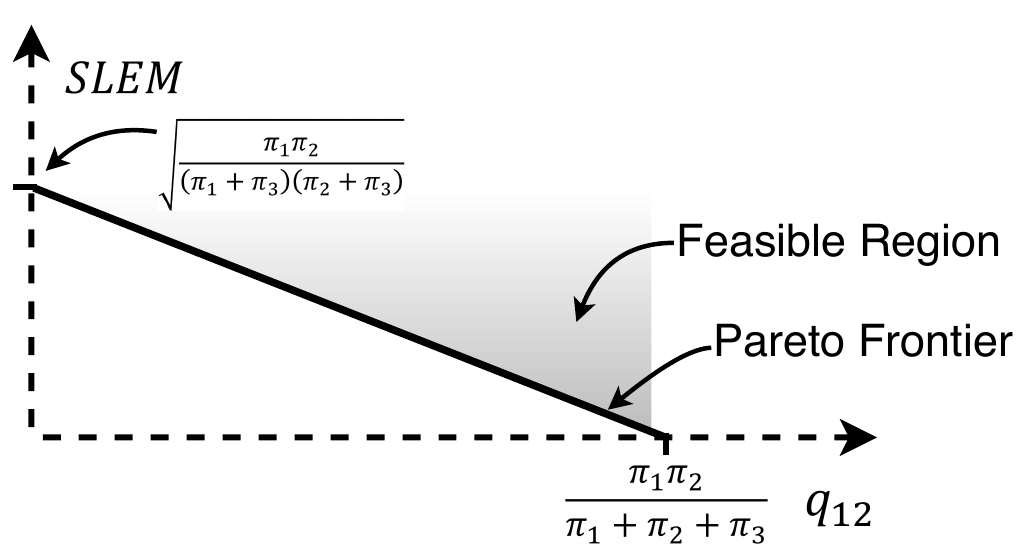}
\caption{}
\label{fig:Friendship-m1-p3geqp1p2-Pareto-Frontier}
\end{subfigure}
\begin{subfigure}[b]{0.4\hsize}
\includegraphics[width=\hsize]{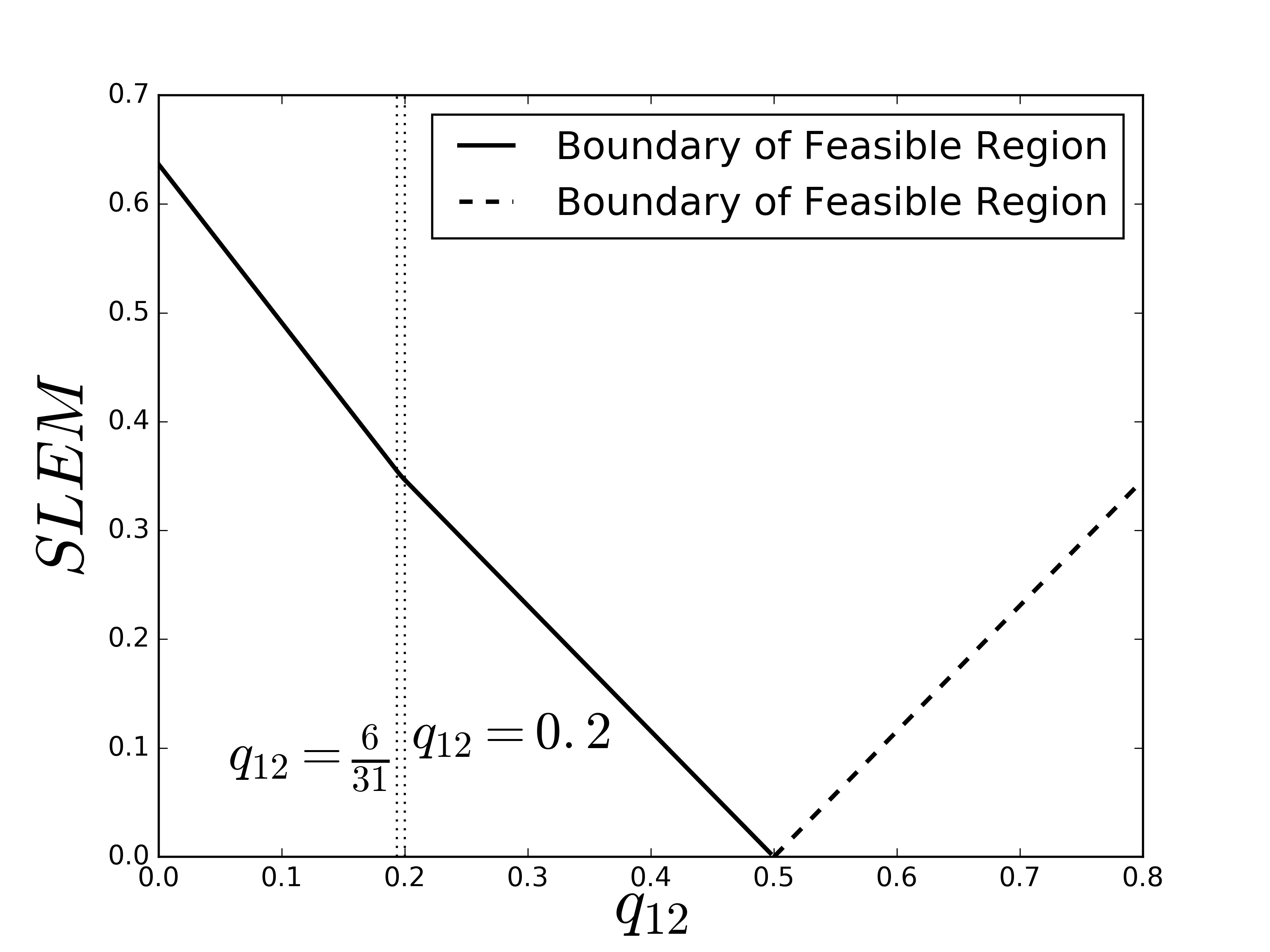}
\caption{}
\label{fig:ParetoExample_3}
\end{subfigure}
\caption{Pareto frontier of friendship graph with $m=1$ blade for (a) $\pi_{3}^{2} \geq \pi_{1} \pi_{2}$ and (b) $\pi_{3}^{2} < \pi_{1} \pi_{2}$ (Example \ref{example:3}).}
\end{figure}
\subsubsection{Equilibrium distributions with $\pi_{3}^{2} < \pi_{1} \pi_{2}$}
For this case, the closed-form formulas for the optimal weights and the $SLEM$ is obtained for three different ranges of $q_{1,2}$.
For $q_{1,2}$ satisfying the following
\begin{equation}
\label{eq:201809302893-main-result}
\begin{gathered}
q_{1,2}  \geq  
\left(  \pi_{1} \pi_{2} - \pi_{3}^{2}  \right) / \left(  2 \pi_{3} + \pi_{1} + \pi_{2}  \right).
\end{gathered}
\end{equation}
the optimal weights and the $SLEM$ are same as (\ref{eq:201809151209-main-result}) and (\ref{eq:201809151244-main-result}).
For the other two ranges of $q_{1,2}$, the closed-form formula for $SLEM$ is cumbersome and we omit presenting it for the general case.
Instead, in the following,
first we address the case of $\pi_{1} = \pi_{2}$, where the three ranges are reduced to two
and we have provided closed-form formulas for the optimal weights and $SLEM$.
Then in an example,
we address the case of $\pi_{1} = 2$, $\pi_{2} = \pi_{3} = 1$, where all three ranges of $q_{1,2}$ are present.

\subsubsection{Case of $\pi_{1} = \pi_{2}$}
\label{sec:pi1=pi2}
For equilibrium distributions with $\pi_{1} = \pi_{2}$,
for
$q_{1,2} \geq \frac{ \pi_{1} - \pi_{3} }{ 2 }$,
the optimal value of $s$, $q_{1,3}$ and $q_{2,3}$ are as below,
\begin{equation}
\label{eq:201809303083-main-result}
\begin{gathered}
s  =  
\left| \pi_{1}^{2} - \left( 2 \pi_{1} + \pi_{3} \right) q_{1,2} \right|   /  \left(   \pi_{1} \left( \pi_{2} + \pi_{3} \right)   \right)
\end{gathered}
\end{equation}
\begin{equation}
\label{eq:201809303092-main-result}
\begin{gathered}
q_{1,3} = q_{2,3}  = 
\left(  \pi_{3} ( \pi_{1} - q_{1,2} )  \right) / \left(  \pi_{1} + \pi_{3}  \right),
\end{gathered}
\end{equation}
These results are obtained by setting $\pi_{1} = \pi_{2}$ in (\ref{eq:201809151209-main-result}), (\ref{eq:201809151244-main-result}),  and (\ref{eq:201809302893-main-result}).
For $q_{1,2} \leq \left( \pi_{1} - \pi_{3} \right) / 2$,
the optimal value of $s$, $q_{1,3}$ and $q_{2,3}$ are 
$q_{1,3} = q_{2,3} = \pi_{3}  /  2$,
\begin{equation}
\label{eq:201809303124-main-result}
\begin{gathered}
s = 
\left(  2 \pi_{1} - \pi_{3} - 4 q_{1,2}  \right) / \left(  2 \pi_{1}  \right).
\end{gathered}
\end{equation}
Note that for $q_{1,2} = \left( \pi_{1} - \pi_{3} \right) / 2$, both values of $s$ in (\ref{eq:201809303083-main-result}) and (\ref{eq:201809303124-main-result}) are equal to $\pi_{3} / ( 2 \pi_{1} )$.
\begin{example}
\label{example:3}
Consider the friendship graph with $m=1$ and the equilibrium distribution
$\pi_{1} = 2$, $\pi_{2} = \pi_{3} = 1$.
For $q_{1,2}$ greater than the limit provided in (\ref{eq:201809302893-main-result}), the optimal $SLEM$ can be derived from (\ref{eq:201809151244-main-result}),
i.e., for $0.2 \leq q_{1,2} \leq 0.5$,
the optimal $SLEM$ is 
$s = \left(  1 - 2 q_{1,2}  \right) / \sqrt{ 3 } $.
For
$  6 / 31   \leq  q_{1,2}  \leq  0.2 $,
the optimal $SLEM$ is 
$  s = \sqrt{  18 q_{1,2}^{2}  -  8 q_{1,2}  +  1  }  $.
For $q_{1,2} \leq  6 / 31  $, the optimal $SLEM$ is 
$  s = \left(    7  -  16 q_{1,2}     \right)  /  11  $.
This curve (as depicted in Fig \ref{fig:ParetoExample_3}) is the Pareto Frontier for this example.
\end{example}
\section{Proof of Main Results}
\label{sec:ProofofMainResults}

In this section, we provide the solution procedure of the results presented in Section \ref{sec:MainResults}.
For the case of $m \geq 2$, based on the interlacing property 
first, the optimization problem corresponding to the friendship graph is reduced to that of the star topology, 
which is then addressed via SDP.
For $m = 1$, the problem is addressed directly via SDP.
\subsection{Interlacing}
To analyze the eigenstructure of the transition probability matrix of the friendship graph, in this subsection we use the interlacing property between the transition probability matrices of the friendship graph and the star topology.
Let $\mathcal{G}_{S}$ be the star topology with $m$ branches of length one, where its vertex and edge sets are denoted by
$\mathcal{V}_{S} = \{ (i) | i=0,1,...,m \}$
and
$\mathcal{E}_{S} = \{ (0,i) | i=1,...,m \}$,
respectively,
and the equilibrium distribution of vertex $i$ is $\widetilde{\pi}_{i}$ for $i=0,...,m$.
The transition probability matrix for the star topology is
$\boldsymbol{P}_{S}  =  \boldsymbol{I} - \boldsymbol{D}_{S}^{-1} \times \boldsymbol{L}_{S} (q)$,
where
$\boldsymbol{D}_{S} = diag( \widetilde{\pi}_{0}, \widetilde{\pi}_{1}, ... \widetilde{\pi}_{m} )$
and
$\boldsymbol{L}_{S}(q)$ is the symmetric Laplacian and it can be written as 
$\boldsymbol{L}_{S} (q)  =  \sum_{i=1}^{m}  \widetilde{q}_{0,i} ( \boldsymbol{e}_{0} - \boldsymbol{e}_{i} )  ( \boldsymbol{e}_{0} - \boldsymbol{e}_{i} )^{T}$, 
where
$\boldsymbol{e}_{j}$ for $j=0,...,m$ are 
column vectors with $1$ in the $i$-th position and zero elsewhere.
Let $\boldsymbol{ \Lambda }$ be a $( 2m+1 ) \times ( m+1 )$ matrix defined as %
$\boldsymbol{ \Lambda } = diag\left( 1, \boldsymbol{J}_{2,1}, ..., \boldsymbol{J}_{2,1} \right)$,
where
$\boldsymbol{J}_{2,1}^{T} = \left[ 1, 1 \right]$
is a $2 \times 1$ matrix with both elements equal to one.
Considering $\boldsymbol{ \Lambda }$ as defined above, by imposing the following relation between the transition probability matrices of 
star topology and 
friendship graph,
\begin{equation}
\label{eq:20181009332}
\begin{gathered}
	\boldsymbol{\Lambda} \times \boldsymbol{P}_{S}
	=
	\boldsymbol{P} \times \boldsymbol{\Lambda}
\end{gathered}
\end{equation}
from %
element-wise comparison of %
both sides of (\ref{eq:20181009332}), we have
\begin{subequations}
\label{eq:20181010391}
\begin{gather}
	\left(  \widetilde{q}_{0,i}  /  \widetilde{\pi}_{0}  \right)
	=
	\left( q_{0,2i-1} + q_{0,2i} \right)  /  \pi_{0} ,
	\label{eq:20181010391b}
	\\
	\left(  \widetilde{q}_{0,i}  /  \widetilde{\pi}_{i}  \right)
	=
	q_{0,2i-1}  /  \pi_{2i-1}  
	=
	q_{0,2i}  /  \pi_{2i}  
	\label{eq:20181010391c}
\end{gather}
\end{subequations}
where
(\ref{eq:20181010391})    %
hold for $i=1,...,m$.
Assuming $\widetilde{\pi}_{0} = \pi_{0}$, from (\ref{eq:20181010391b}), we have
$ \widetilde{q}_{0,i} = q_{0,2i-1} + q_{0,2i}$, 
where considering (\ref{eq:20181010391c}), 
$\widetilde{\pi}_{i}   =   \pi_{2i-1}  +  \pi_{2i}$ 
is obtained.
From (\ref{eq:20181010391c}), it is obvious that the ratios
$\frac{ q_{0,2i-1} }{ \pi_{2i-1} }$
and
$\frac{ q_{0,2i} }{ \pi_{2i} }$
are the same where we denote it by $\mu_{i}$, i.e.,
$\mu_{i}  =  \left( q_{0,2i-1} / \pi_{2i-1} \right)  =  \left( q_{0,2i} / \pi_{2i} \right)  =  \left(  \widetilde{q}_{0,i}  /  \widetilde{\pi}_{i}  \right)$
for $i=1,...,m$.
\begin{lemma}
	\label{sec:Lemma-Star_Friendship-Eigenstructure}
	All eigenvalues of the star graph that satisfies (\ref{eq:20181009332}) and
	the constraint $\widetilde{\pi}_{0} = \pi_{0}$    %
	are among the eigenvalues of the friendship graph.
\end{lemma}
Proof of this lemma is provided in Appendix \ref{sec:Appendix-Lemma-Proof-Star_Friendship-Eigenstructure}.
Consider the set of column vectors
$\{ \boldsymbol{e}_{i} | i = 1, ..., m \}$    %
with $1$ in the index corresponding to vertex $(i) \in \mathcal{V}$ and zero elsewhere.
We define the new set of orthogonal basis as below,
\begin{equation}
\label{eq:20181011335}
\begin{gathered}
	\boldsymbol{e}_{0}  =  \boldsymbol{e}_{0}, 
	\;\;\;\;\;
		\boldsymbol{e}_{i,\mp}  =
		(  \boldsymbol{e}_{2i-1}  /  \sqrt{2}  )
		\mp  
		(  \boldsymbol{e}_{2i}  /  \sqrt{2} ).
\end{gathered}
\end{equation}
The transition probability matrix of the friendship graph in the new basis (\ref{eq:20181011335}) takes 
the block diagonal form 
$\boldsymbol{P}  =  diag\left(  \boldsymbol{P}_{0},  \boldsymbol{P}_{1}, ..., \boldsymbol{P}_{m}  \right)$,
where
$\boldsymbol{P}_{i}$    %
for $i=1,...,m$ are $1 \times 1$ blocks as below,
\begin{equation}
\label{eq:20181011367}
\begin{gathered}
	1 - \mu_{i} - 
		\left( q_{2i-1,2i} \left( \pi_{2i-1} + \pi_{2i} \right) \right) / ( \pi_{2i-1}  \pi_{2i} )
\end{gathered}
\end{equation}
and the $(m+1) \times (m+1)$ block
$\boldsymbol{P}_{0}$    %
is as below,
\begin{equation}
	\label{eq:20181011419}
	\begin{gathered}
		\boldsymbol{P}_{0} =
		\boldsymbol{I} - 
		\sum\nolimits_{i=1}^{m} 
		\left( \iota_{i} \boldsymbol{e}_{0} - \mu_{i} \boldsymbol{e}_{i} \right)
		\left( \boldsymbol{e}_{0} - \boldsymbol{e}_{i} \right)^{T},
	\end{gathered}
\end{equation}
where 
$\iota_{i} = \frac{ \mu_{i} ( \pi_{2i-1} + \pi_{2i} ) }{ \pi_{0} }$ 
for $i=1,...,m$.
The $1 \times 1$ blocks in (\ref{eq:20181011367}) are the eigenvalues of the transition probability matrix
$\boldsymbol{P}_{0}$    %
as well.
We denote these eigenvalues by $\overline{s}_i$, as below,
\begin{equation}
	\label{eq:20181013903}
	\begin{gathered}
		\overline{s}_{i} = 1 - \mu_{i} 
		- 
			\left( q_{2i-1,2i} \left( \pi_{2i-1} + \pi_{2i} \right) \right) / \left( \pi_{2i-1}  \pi_{2i} \right)
	\end{gathered}
\end{equation}
for $i=1,...,m$.
	Lemma \ref{sec:Lemma-Star_Friendship-Eigenstructure} shows that there is interlacing \cite{HAEMERS1995593} (\ref{sec:AppendixInterlacing}) between the eigenvalues of the transition probability matrices of the Reversible Markov chains over the friendship graph and its corresponding star topology.
	But this interlacing is not necessarily tight interlacing (Appendix \ref{sec:AppendixInterlacing}).
	Considering the $SLEM$ of $\boldsymbol{P}_{0}$, given in (\ref{eq:20181011419}) and the single eigenvalues $\overline{s}_{i}$ given in (\ref{eq:20181013903}),
	it is obvious that if
	$|\overline{s}_{i}| \leq SLEM\left(  \boldsymbol{P}_{0}  \right) $, then the aforementioned interlacing is a tight interlacing.
	This is an important advantage in tackling
		problem (\ref{eq:Eq20181027427}),
	since
	it shows that the $SLEM$ of the transition probability matrix $\boldsymbol{P}$ is equal to that of $\boldsymbol{P}_{0}$ given in (\ref{eq:20181011419}).
	Hence, instead of minimizing the $SLEM$ of the whole transition probability matrix $\boldsymbol{P}$, it suffices to minimize the $SLEM$ of $\boldsymbol{P}_{0}$ (\ref{eq:20181011419}) with the constraint $|\overline{s}_{i}| \leq SLEM\left(  \boldsymbol{P}_{0}  \right) $.
	Thus, instead of the optimization problem (\ref{eq:Eq201712231257}),
	the following optimization problem has to be addressed,
	\begin{equation}
		\label{eq:Eq201811023654}
		\begin{aligned}
			\hspace{-10pt}
			\min\limits_{ \{ q_{i,j} | \{i,j\} \in \mathcal{E}_{V} \} }
			\;\;
			&\mu\left( \boldsymbol{P}_{0} \right),
			\\
			s.t.
			\qquad\;\;
			&
			\mu_{i} \geq  0, \;
			1 - \mu_{i} \geq  0,
			\; \text{for} \; i = 1, ..., m,
			\\ &
				\widehat{\boldsymbol{\pi}}^{T} \boldsymbol{P}_{0}  =  \widehat{\boldsymbol{\pi}}^{T},
				\;\; %
				\boldsymbol{P}_{0} \boldsymbol{1}
				=
				\boldsymbol{1},
				\;\;
				|\overline{s}_{i}| \leq \mu \left(  \boldsymbol{P}_{0}  \right)
			\\
			&
			1 - \sum\nolimits_{k=1}^{m} 
				\left( \mu_{k} ( \pi_{2k-1} + \pi_{2k} ) / \pi_{0} \right) 
			\geq 0,
		\end{aligned}
	\end{equation}
	where
		$\widehat{\boldsymbol{\pi}}^{T} = [ 1, \frac{ \pi_{1} + \pi_{2} }{ \pi_{0} }, ..., \frac{ \pi_{2m-1} + \pi_{2m} }{ \pi_{0} } ]$.
	Note that the conclusion above holds true for $m \geq 2$.
	Now
	we formulate the optimization problem (\ref{eq:Eq201811023654}) in the form of semidefinite programming problem  \cite{BoydConvexBook,JafarizadehIEEESensors2011}.
In doing so, we have to convert the $\boldsymbol{P}_{0}$ into a symmetric matrix, which can be done by
multiplying $\boldsymbol{P}_{0}$ from left and right with the $(m+1) \times (m+1)$ diagonal matrix
	$\boldsymbol{M}$ $=$ $diag($ $1$, $\sqrt{ \frac{ \pi_{1} + \pi_{2} }{ \pi_{0} } }$, $...$,  $\sqrt{ \frac{ \pi_{2m-1} + \pi_{2m} }{ \pi_{0} } }$ $)$.
By doing so, $\widehat{\boldsymbol{P}}_{0}$ is obtained as below,
\begin{equation}
	\label{eq:Eq201808108358}
	\begin{gathered}
		\widehat{\boldsymbol{P}}_{0}    %
		= \boldsymbol{I} -
		\sum\nolimits_{i=1}^{m} \mu_{i}
		\left( \vartheta_{i} \boldsymbol{e}_{0} - \boldsymbol{e}_{i} \right)
		\left( \vartheta_{i} \boldsymbol{e}_{0} - \boldsymbol{e}_{i} \right)^{T},
	\end{gathered}
\end{equation}
where
$\vartheta_{i} = \sqrt{ \frac{ \pi_{2i-1} + \pi_{2i} }{ \pi_{0} } }$,
for $i=1,...,m$.
Note that both matrices presented in (\ref{eq:20181011419}) and
(\ref{eq:Eq201808108358}),    %
have the same set of eigenvalues.
Therefore, we use
(\ref{eq:Eq201808108358})    %
for analyzing the eigenstructure of
$\boldsymbol{P}_{0}$.    %
Since the star topology is a tree graph, the vectors
$\{  \vartheta_{i}  \boldsymbol{e}_{0} - \boldsymbol{e}_{i} | i=1,...,m \}$
are independent of each other and thus they can form a basis.

Following a procedure similar to that of \cite{BoydFastestmixing2003}, it can be shown that
	problem (\ref{eq:Eq201811023654})
is a convex optimization problem
and
	considering  (\ref{eq:Eq201808108358}), it
can be derived as the following semidefinite programming problem,
	\begin{subequations}
		\label{eq:Eq201801262501}
		\begin{align}
			\min_{    \resizebox{.16\hsize}{!}{$\{ \mu_{1}, ..., \mu_{m} \} \cup \{ s \} $}    }
			\;\;
			&s,
			\nonumber
			\\
			s.t.
			\quad
			&
				(s - 1) \boldsymbol{I}
				+
				\sum\nolimits_{i=1}^{m} \mu_{i}
				\left( \vartheta_{i} \boldsymbol{e}_{0} - \boldsymbol{e}_{i} \right)
				\left( \vartheta_{i} \boldsymbol{e}_{0} - \boldsymbol{e}_{i} \right)^{T}
				+
				\boldsymbol{v}_{0} \boldsymbol{v}_{0}^{T}
				\succcurlyeq    %
				0
			\label{eq:Eq201801262501a}
			\\
			&
				(s + 1) \boldsymbol{I}
				-
				\sum\nolimits_{i=1}^{m} \mu_{i}
				\left( \vartheta_{i} \boldsymbol{e}_{0} - \boldsymbol{e}_{i} \right)
				\left( \vartheta_{i} \boldsymbol{e}_{0} - \boldsymbol{e}_{i} \right)^{T}
				-
				\boldsymbol{v}_{0} \boldsymbol{v}_{0}^{T}
				\succcurlyeq    %
				0
			\label{eq:Eq201801262501b}
			\\
			&
					s   -   1 + \mu_{i} + 
					( q_{2i-1,2i} \left( \pi_{2i-1} + \pi_{2i} \right) / ( \pi_{2i-1}  \pi_{2i} ) )
					\geq  0
			\label{eq:Eq201801262501c}
			\\&
					s   +   1 - \mu_{i} - 
					( q_{2i-1,2i} \left( \pi_{2i-1} + \pi_{2i} \right) / ( \pi_{2i-1}  \pi_{2i} ) )
					\geq  0
			\label{eq:Eq201801262501d}
			\\&
					1 - \mu_{i} - 
					\frac{  q_{2i-1,2i}  }{  \pi_{2i-1}   }
					\geq 0,
					\;\;\; %
					1 - \mu_{i} - 
					\frac{  q_{2i-1,2i}  }{  \pi_{2i}  }
					\geq 0
			\label{eq:Eq201801262501f}
			\\&
				1 - \sum\nolimits_{i=1}^{m} \mu_{i} \vartheta_{i}^{2} \geq 0
			\label{eq:Eq201801262501g}
		\end{align}
	\end{subequations}
	where $\boldsymbol{v}_{0}^{T} = [ 1, \vartheta_{1}, ..., \vartheta_{m} ]$, and the constraints (\ref{eq:Eq201801262501c}), (\ref{eq:Eq201801262501d}), 
	(\ref{eq:Eq201801262501f}) hold for $i = 1, ..., m$.
	The symbol $\preccurlyeq$ denotes matrix inequality, i.e., $\boldsymbol{X} \preccurlyeq \boldsymbol{Y}$ means $\boldsymbol{Y} - \boldsymbol{X}$ is positive semidefinite.
Introducing
$\boldsymbol{x} = [ \mu_{1}, ..., \mu_{m}, s]^{T}$,
$\boldsymbol{c} = [ 0, ..., 0, 1]^{T}$,
and the block diagonal matrices
$\boldsymbol{F}_{0}$ and $\boldsymbol{F}_{i}$
with the following blocks
\begin{equation}
	\nonumber
	\begin{gathered}
		\boldsymbol{F}_{0}  =  diag(
		\boldsymbol{v}_{0} \boldsymbol{v}_{0}^{T} - \boldsymbol{I},
		\boldsymbol{I} - \boldsymbol{v}_{0} \boldsymbol{v}_{0}^{T},
		\boldsymbol{D}_{1},
		-\boldsymbol{D}_{1},
		\boldsymbol{D}_{2},
		\boldsymbol{D}_{3},
		1
		)
	\end{gathered}
\end{equation}
where
$\boldsymbol{D}_{1}$, $\boldsymbol{D}_{2}$ and $\boldsymbol{D}_{3}$ are $m \times m$ diagonal matrices where
the $i$-th diagonal elements of $\boldsymbol{D}_{1}$ and $\boldsymbol{D}_{2}$ and $\boldsymbol{D}_{3}$ are %
$\frac{ q_{2i-1,2i} \left( \pi_{2i-1} + \pi_{2i} \right) }{ \pi_{2i-1}  \pi_{2i} }  -  1$,
and
$1$  $-$  $\frac{  q_{2i-1,2i}  }{ \pi_{2i-1} }$,
and
$1  -  \frac{  q_{2i-1,2i}  }{ \pi_{2i} }$,
respectively,
for $i=1,...,m$,
\begin{equation}
	\nonumber
	\begin{gathered}
		\boldsymbol{F}_{s} = diag( \boldsymbol{I}_{m+1}, \boldsymbol{I}_{m+1}, \boldsymbol{I}_{m}, \boldsymbol{I}_{m}, \boldsymbol{0}_{m}, \boldsymbol{0}_{m}, 0 ),
	\end{gathered}
\end{equation}
$\boldsymbol{F}_{i}$ $=$ $diag($    $( \vartheta_{i} \boldsymbol{e}_{0} - \boldsymbol{e}_{i} )$ $( \vartheta_{i} \boldsymbol{e}_{0} - \boldsymbol{e}_{i} )^{T}$,
$- ( \vartheta_{i} \boldsymbol{e}_{0} - \boldsymbol{e}_{i} )$ $( \vartheta_{i} \boldsymbol{e}_{0} - \boldsymbol{e}_{i} )^{T}$,
$\boldsymbol{e}_{i} \boldsymbol{e}_{i}^{T}$,
$- \boldsymbol{e}_{i} \boldsymbol{e}_{i}^{T}$,
$- \boldsymbol{e}_{i} \boldsymbol{e}_{i}^{T}$,
$- \boldsymbol{e}_{i} \boldsymbol{e}_{i}^{T}$,
$- \vartheta_{i}^{2} )$
for $i = 1, ..., |\boldsymbol{x}|-1$,
problem (\ref{eq:Eq201801262501}) can be written in the standard form of the semidefinite programming \cite{BoydConvexBook,JafarizadehIEEESensors2011} as below,
\begin{equation}
\label{eq:Eq201801272653}
\begin{aligned}
\min\limits_{\boldsymbol{x}}
\;\;
&\boldsymbol{c}^{T} \cdot \boldsymbol{x},
\\
s.t.
\quad
&
\boldsymbol{F}(x) = \sum\nolimits_{i=1}^{|\boldsymbol{x}|} \boldsymbol{x}_{i} \boldsymbol{F}_{i} + \boldsymbol{F}_{0} \succeq 0
\end{aligned}
\end{equation}
The dual problem is as following,
\begin{equation}
\label{eq:Eq201801272673}
\begin{aligned}
\max\limits_{\boldsymbol{Z}}
\;\;
&-Tr \left[ \boldsymbol{F}_{0} \cdot \boldsymbol{Z} \cdot \boldsymbol{Z}^{T} \right],
\\
s.t.
\quad
&
Tr \left[ \boldsymbol{F}_{s} \cdot \boldsymbol{Z} \cdot \boldsymbol{Z}^{T} \right]  = \boldsymbol{c}_{|\boldsymbol{x}|} = 1,
\\
&
	Tr \left[ \boldsymbol{F}_{i} \cdot \boldsymbol{Z} \cdot \boldsymbol{Z}^{T} \right]  = \boldsymbol{c}_{i} = 0 \;\; \text{for} \;\; i=1, ..., |\boldsymbol{x}|-1.
\end{aligned}
\end{equation}
The dual variable $\boldsymbol{Z}$ can be written as
$\boldsymbol{Z}$  $=$
$[ \boldsymbol{Z}_{1}^{T}$, $\boldsymbol{Z}_{2}^{T}$, $\boldsymbol{Z}_{3}^{T}$, $\boldsymbol{Z}_{4}^{T}$, $\boldsymbol{Z}_{5}^{T}$, $\boldsymbol{Z}_{6}^{T}$, $\boldsymbol{Z}_{7}^{T} ]^{T}$
where
$\boldsymbol{Z}_{1}$  $=$  $\sum\nolimits_{ i=1 }^{ m } a_{i}$ $($ $\vartheta_{i}$ $\boldsymbol{e}_{0}$  $-$  $\boldsymbol{e}_{i}$ $)$,
$\boldsymbol{Z}_{2}$  $=$  $\sum\nolimits_{ i=1 }^{ m }$ $b_{i}$ $($ $\vartheta_{i}$ $\boldsymbol{e}_{0}$ $-$ $\boldsymbol{e}_{i}$ $)$,
$\boldsymbol{Z}_{3}$  $=$  $\sum\nolimits_{ i=1 }^{ m }$ $c_{i}$ $\boldsymbol{e}_{i}$,
$\boldsymbol{Z}_{4}$  $=$  $\sum\nolimits_{ i=1 }^{ m }$ $d_{i}$ $\boldsymbol{e}_{i}$,
$\boldsymbol{Z}_{5}$  $=$  $\sum\nolimits_{ i=1 }^{ m }$ $e_{i}$ $\boldsymbol{e}_{i}$,
$\boldsymbol{Z}_{6}$  $=$  $\sum\nolimits_{ i=1 }^{ m }$ $f_{i}$ $\boldsymbol{e}_{i}$,
$\boldsymbol{Z}_{7}  =  g $
Using
these definitions above,
the dual constraints in
(\ref{eq:Eq201801272673})  %
(i.e. $Tr\left( \boldsymbol{Z} \boldsymbol{F}_{i} \boldsymbol{Z}^{T} \right)  =  0$)
reduce to the following
\begin{equation}
\label{eq:Eq20171114916}
\begin{gathered}
		\left( \boldsymbol{Z}_{1}^{T} \left( \vartheta_{i} \boldsymbol{e}_{0}  -  \boldsymbol{e}_{i} \right) \right)^{2}
		-
		\left( \boldsymbol{Z}_{2}^{T} \left( \vartheta_{i} \boldsymbol{e}_{0}  -  \boldsymbol{e}_{i} \right) \right)^{2}
		+
		\left( \boldsymbol{Z}_{3}^{T} \boldsymbol{e}_{i} \right)^{2}
		-
		\left( \boldsymbol{Z}_{4}^{T} \boldsymbol{e}_{i} \right)^{2}
		-
		\left( \boldsymbol{Z}_{5}^{T} \boldsymbol{e}_{i} \right)^{2}
		-
		\left( \boldsymbol{Z}_{6}^{T} \boldsymbol{e}_{i} \right)^{2}
		-
		\left( \boldsymbol{Z}_{7}^{T} \vartheta_{i} \right)^{2}
		=    0
\end{gathered}
\end{equation}
The complementary slackness condition
\cite{BoydConvexBook,JafarizadehIEEESensors2011},  %
reduces to
\begin{subequations}
\label{eq:Eq20171114938}
\begin{gather}
		\left( (s - 1) \boldsymbol{I}
		+
		\sum\nolimits_{j=1}^{m} \mu_{j}
		\left( \vartheta_{j} \boldsymbol{e}_{0} - \boldsymbol{e}_{j} \right)
		\left( \vartheta_{j} \boldsymbol{e}_{0} - \boldsymbol{e}_{j} \right)^{T}
		\right) \boldsymbol{Z}_{1}  =  0
	\label{eq:Eq20171114938a}
	\\
		\left(  (s + 1) \boldsymbol{I}
		-
		\sum\nolimits_{j=1}^{m} \mu_{i}
		\left( \vartheta_{j} \boldsymbol{e}_{0} - \boldsymbol{e}_{j} \right)
		\left( \vartheta_{j} \boldsymbol{e}_{0} - \boldsymbol{e}_{j} \right)^{T}  \right)  \boldsymbol{Z}_{2} = 0
	\label{eq:Eq20171114938b}
	\\
			c_{i} \left(   s   -   1 + \mu_{i} + 
			 q_{2i-1,2i} \frac{ \pi_{2i-1} + \pi_{2i} }{ \pi_{2i-1}  \pi_{2i} } 
			\right)  =  0
	\label{eq:Eq20171114938c}
	\\
			d_{i} \left(   s   +   1 - \mu_{i} - q_{2i-1,2i} \frac{ \pi_{2i-1} + \pi_{2i} }{ \pi_{2i-1}  \pi_{2i} }  \right)  =  0
	\label{eq:Eq20171114938d}
	\\
		e_{i} \left(   1 - \mu_{i} - \frac{  q_{2i-1,2i}  }{  \pi_{2i-1}  }   \right)  =  0
	\label{eq:Eq20171114938e}
	\\
		f_{i} \left(   1 - \mu_{i} - \frac{ q_{2i-1,2i} }{ \pi_{2i} }   \right)  =  0
	\label{eq:Eq20171114938f}
	\\
	g \left(   1 - \sum\nolimits_{j=1}^{m} \mu_{j} \vartheta_{j}^{2}   \right)  =  0
	\label{eq:Eq20171114938g}
\end{gather}
\end{subequations}
for $i=1,...,m$.
Considering (\ref{eq:Eq20171114938a}) and (\ref{eq:Eq20171114938b})
and the linear independence of the vectors
$\{  \vartheta_{i}  \boldsymbol{e}_{0} - \boldsymbol{e}_{i} | i=1,...,m \}$,
equation (\ref{eq:Eq20171114916}) can be written as below,
\begin{equation}
	\label{eq:Eq201810252968}
	\begin{gathered}
			\left(  \left( s - 1 \right) 
			\frac{ a_{i} }{ \vartheta_{i} \mu_{i} }  
			\right)^{2}    %
			-
			\left(  \left( s + 1 \right) 
			\frac{ b_{i} }{ \vartheta_{i} \mu_{i} }  
			\right)^{2}    %
			+
			c_{i}^{2} %
		-
		d_{i}^{2} %
		-
		e_{i}^{2} %
		-
		f_{i}^{2} %
		-
		\left( g \vartheta_{i} \right)^{2} %
		=  0
	\end{gathered}
\end{equation}
\begin{remark}
From (\ref{eq:Eq20171114916}), it is obvious that
$\boldsymbol{Z}_{1} = \boldsymbol{0}$
	and $\boldsymbol{Z}_{3} = \boldsymbol{0}$
result
	$\boldsymbol{Z} = 0$
which is not acceptable.
Thus, it can be concluded that
	both
$\boldsymbol{Z_{1}}$
	and $\boldsymbol{Z}_{3}$
is always non-zero.
\end{remark}
\begin{theorem}
If $\boldsymbol{Z_{1}}$ and $\boldsymbol{Z_{2}}$ are non-zero,
then
for optimal value of $s$,
$1 - s$ and $s + 1$ are both eigenvalues of
	$\widehat{\boldsymbol{P}}_{0}$,
with corresponding eigenvectors $\boldsymbol{Z_{1}}$ and $\boldsymbol{Z_{2}}$.
\end{theorem}
\begin{proof}
The constraints of the optimization problem
(\ref{eq:Eq201801262501}) %
state that all eigenvalues of
	$\widehat{\boldsymbol{P}}_{0}$
are bounded in between $s-1$ and $1+s$.
Since both $\boldsymbol{Z_{1}}$ and $\boldsymbol{Z_{2}}$ are non-zero,
then for the optimal $SLEM$, it can be concluded that
	$s$ $=$ $\mu$ $(  \widehat{\boldsymbol{P}}_{0}  )$ $=$ $1$ $-$ $\lambda_{2}$ $(  \widehat{\boldsymbol{P}}_{0}  )$ $=$ $\lambda_{N}$ $(  \widehat{\boldsymbol{P}}_{0}  )$ $-$ $1$.
\end{proof}
If the constraints of (\ref{eq:Eq201801262501c}) to (\ref{eq:Eq201801262501g})
are strict (i.e. $- \sum\nolimits_{k \neq i} q_{i,k} + \pi_{i} > 0$ for $i=0, ..., m$
and $|\overline{s}_{i}| <  s$)
then it can be concluded that $\pi_{i} - \sum_{j \neq i} q_{i,j} \neq 0$.
Thus, we have
\begin{equation}
\label{eq:Eq20171114955}
\begin{gathered}
	c_{i}
		= d_{i} = e_{i} = f_{i} = g
	=  0
	\;\;  \text{for}  \;\;  i=1, ..., m.
\end{gathered}
\end{equation}
This will dictate additional constraints on the equilibrium distribution.
In general (\ref{eq:Eq20171114938}) hold for any underlying topology.
But in the case of tree topologies where vectors
$\{  \vartheta_{i}  \boldsymbol{e}_{0} - \boldsymbol{e}_{i} | i=1,...,m \}$
are independent of each other, (\ref{eq:Eq20171114938a}) and (\ref{eq:Eq20171114938b}) can be interpreted as below,
\begin{subequations}
\label{eq:Eq20171114964}
\begin{gather}
		\left( s - 1 \right) a_{i}  +  \mu_{i} \left( \vartheta_{i} \boldsymbol{e}_{0} - \boldsymbol{e}_{i} \right)^{T} \boldsymbol{Z}_{1}  =  0,
	\label{eq:Eq20171114964a}
	\\
		\left( s + 1 \right) b_{i}  -  \mu_{i} \left( \vartheta_{i} \boldsymbol{e}_{0} - \boldsymbol{e}_{i} \right)^{T} \boldsymbol{Z}_{2}  =  0,
	\label{eq:Eq20171114964b}
\end{gather}
\end{subequations}
for $i=1,...,m$.
From
(\ref{eq:Eq201810252968})
and
(\ref{eq:Eq20171114955}),
we have
\begin{equation}
\label{eq:Eq201807237745}    %
\begin{gathered}
		\left( (s-1) 
		\frac{ a_{i} }{ \vartheta_{i} \mu_{i} } 
		\right)^{2}  =  \left( (s+1)  
		\frac{ b_{i} }{ \vartheta_{i} \mu_{i} } 
		\right)^{2}
\end{gathered}
\end{equation}
for $i=1,...,m$.
Considering the Gram matrix 
$\boldsymbol{G}_{ij} = ( \vartheta_{i}  \boldsymbol{e}_{0} - \boldsymbol{e}_{i} )^{T}  (  \vartheta_{j}  \boldsymbol{e}_{0} - \boldsymbol{e}_{j} )$, 
which is equal to 
$\vartheta_{i}^{2} + 1$ if $i=j$
and 
$\vartheta_{i} \vartheta_{j}$ otherwise,
(\ref{eq:Eq20171114964}) can be written as
\begin{subequations}
	\label{eq:Eq201807237825}
	\begin{gather}
		a_{i} \mu_{i} + \mu_{i}  \vartheta_{i}  \sum\nolimits_{j=1}^{m}  \vartheta_{j}  a_{j}   =  ( 1 - s ) a_{i},
		\label{eq:Eq201807237825a}
		\\
		b_{i} \mu_{i} + \mu_{i} \vartheta_{i} \sum\nolimits_{j=1}^{m}  \vartheta_{j} b_{j}   =  ( 1 + s ) b_{i},
		\label{eq:Eq201807237825b}
	\end{gather}
\end{subequations}
for $i=1,...,m$,
Since
$\sum\nolimits_{j=1}^{m} \vartheta_{j} a_{j}$
and
$\sum\nolimits_{j=1}^{m} \vartheta_{j} b_{j}$
are independent of $i$, thus the following can be concluded from (\ref{eq:Eq201807237825})
\begin{subequations}
	\label{eq:Eq201807237841}
	\begin{gather}
			\left( \mu_{i} + s - 1 \right)  
			\frac{ a_{i} }{ \mu_i \vartheta_{i} }
			= %
			\left( \mu_{j} + s - 1 \right)  
			\frac{ a_{j} }{ \mu_j \vartheta_{j} },
		\label{eq:Eq201807237841a}
		\\
			\left( \mu_{i} - s - 1 \right)  
			\frac{ b_{i} }{ \mu_i \vartheta_{i} }
			= %
			\left( \mu_{j} - s - 1 \right)  
			\frac{ b_{j} }{ \mu_j \vartheta_{j} },
		\label{eq:Eq201807237841b}
	\end{gather}
\end{subequations}
for $i=1,...,m$,
It can be shown that 
considering (\ref{eq:Eq201807237745}), from (\ref{eq:Eq201807237841}), the following is the only acceptable conclusion
$\mu_{i} = \mu_{j}$ for $i,j=1,...,m$.
Defining
$\mu = \mu_{i}$,
equations (\ref{eq:Eq201807237825}) can be written as below,
\begin{subequations}
\label{eq:Eq201807237897}
\begin{gather}
	a_{i} \mu + \mu \vartheta_{i}  \sum\nolimits_{j=1}^{m}  \vartheta_{j} a_{j}   =  ( 1 - s ) a_{i},
	\label{eq:Eq201807237897a}
	\\
	b_{i} \mu + \mu  \vartheta_{i}  \sum\nolimits_{j=1}^{m}   \vartheta_{j}  b_{j}   =  ( 1 + s ) b_{i},
	\label{eq:Eq201807237897b}
\end{gather}
\end{subequations}
for $i=1,...,m$.
Multiplying both sides of (\ref{eq:Eq201807237897}) with $ \vartheta_{i}$ and
summing the resultant equations over $i=1,...,m$, we have
\begin{subequations}
\label{eq:Eq201807237929}
\begin{gather}
	\left(
	\mu \left( 1 + \sum\nolimits_{i=1}^{m} \vartheta_{i}^{2}  \right)  -  \left( 1 - s \right)
	\right)
	\sum\nolimits_{i=1}^{m}  \vartheta_{i}  a_{i}
	=
	0,
	\label{eq:Eq201807237929a}
	\\
	\left(
	\mu \left( 1 + \sum\nolimits_{i=1}^{m} \vartheta_{i}^{2}  \right)  -  \left( 1 + s \right)
	\right)
	\sum\nolimits_{i=1}^{m}  \vartheta_{i}  b_{i}
	=
	0,
	\label{eq:Eq201807237929b}
\end{gather}
\end{subequations}
Of all possible answers of (\ref{eq:Eq201807237929}), the one corresponding to $\sum_{i=1}^{m}  \vartheta_{i}  a_{i}  = 0$ and $\sum_{i=1}^{m}  \vartheta_{i}  b_{i}  \neq 0$ is acceptable, where from (\ref{eq:Eq201807237897a}) and (\ref{eq:Eq201807237929b}), we have
$\mu = 1 - s$,
and 
$\mu \left( 1 + \sum_{i=1}^{m} \vartheta_{i}^{2}  \right) = 1 + s$,
which in turn the following can be concluded for $\mu$,
\begin{equation}
\label{eq:Eq201807237973} 
\begin{gathered} 
		\mu 
		= 
		\frac{ 2 }{ 2  + \sum_{i=1}^{m} \vartheta_{i}^{2} } 
		= 
			\frac{ 2 \pi_{0} }{ 2 \pi_{0}  + \sum\nolimits_{i=1}^{m} \left( \pi_{2i-1} + \pi_{2i} \right) },
\end{gathered}
\end{equation}
and thus (\ref{eq:Eq201807237981}).
Substituting (\ref{eq:Eq201807237973}) in $\mu = 1 - s$, equation (\ref{eq:Eq201807237991}) can be concluded for $s$.
Thus, (\ref{eq:Eq201807237973}), (\ref{eq:Eq201807237981}) and (\ref{eq:Eq201807237991}) are the only acceptable answers from (\ref{eq:Eq201807237929}).
And they
hold true if
(\ref{eq:Eq20171114955})
is satisfied
i.e.,
$\pi_{0}$ $\geq$ $\sum_{i=1}^{m}$ $q_{0,2i-1}$ $+$ $q_{0,2i} $.
Considering (\ref{eq:Eq201807237981}), this constraint can be written as
$\sum_{i=1}^{m}$ $($ $\pi_{2i-1}$ $+$ $\pi_{2i}$ $)$ $\leq$ $2\pi_{0}$
Thus, it can be concluded
that for 
$\sum_{i=1}^{m}$ $($ $\pi_{2i-1}$ $+$ $\pi_{2i}$ $)$ $\leq$ $2\pi_{0}$, 
we have $SLEM \leq \frac{1}{2}$.
For
$\sum_{i=1}^{m}$ $($ $\pi_{2i-1}$ $+$ $\pi_{2i}$ $)$ $>$ $2 \pi_{0}$,
and assuming $b_{i} = 0$ for $i$ $=$ $1$,...,$m$,
and from the fact that $c_{i}=0$ for $i$ $=$ $1$,...,$m$, it can be concluded that
$\left( \left( s - 1 \right) \frac{ a_{i} }{ \vartheta_{i} \mu_i } \right)^{2} = g^{2}$,
and thus
$( \frac{ a_{i} }{ \vartheta_{i} \mu_{i} } )^{2}$ $=$ $( \frac{ a_{j} }{ \vartheta_{j} \mu_{j} } )^{2}$
for $i,j=1,...,m$.
Hence, from (\ref{eq:Eq201807237841a}), we have
$( \mu_{i}$ $+$ $s$ $-$ $1)$ 		$=$		$\pm$ $(  \mu_{j}$ $+$ $s$ $-$ $1  )$
for $i,j=1,...,m$,
the possible solutions are
$\mu_{i} = \mu_{j}$, for $i,j=1,...,m$.
and %
$s = 1 - \mu_{i} - \mu_{j}$ for $i$ $\neq$ $j$ $=$ $1$,...,$m$, 
where for $m > 2$, 
$\mu_{i} = \mu_{j}$    %
holds true.
Defining
$\mu = \mu_{i}$
for $i=1,...,m$, from
$\sum_{i=1}^{m} ( q_{0,2i-1} + q_{0,2i} ) = \pi_{0}$,
we have
$\pi_{0} = \mu \sum_{i=1}^{m} ( \pi_{2i-1} + \pi_{2i}  )$
and thus
$\mu = \frac{ \pi_{0} }{  \sum\limits_{i=1}^{m} ( \pi_{2i-1} + \pi_{2i} )  }$,
which in turn results in (\ref{eq:Eq201807278235}).
From (\ref{eq:Eq201807237929a}), the only acceptable answer is obtained for $\sum_{i=1}^{m} \vartheta_{i} a_{i} = 0$, where from (\ref{eq:Eq201807237897a}), we have $ s = 1 - \mu$, or equivalently (\ref{eq:Eq201807278283}), which is positive for 
$\sum_{j=1}^{m} ( \pi_{2j-1} + \pi_{2j} ) > 2 \pi_{0}$.
From (\ref{eq:Eq201807278283}), it is obvious that for 
$\sum_{j=1}^{m} ( \pi_{2j-1} + \pi_{2j} ) > 2 \pi_{0}$, 
the resultant $SLEM$ is greater than $\frac{1}{2}$.
Note that the results in (\ref{eq:Eq201807237981}), (\ref{eq:Eq201807237991}), (\ref{eq:Eq201807278235}) and (\ref{eq:Eq201807278283}) hold for $m \geq 3$.
From the constraints
\begin{equation}
\label{eq:corner-constraints-star-raw}
\begin{gathered}
	q_{2i-1,2i} + q_{0,2i-1}  \leq  \pi_{2i-1},
	\;\;\;\; \;\;  %
	q_{2i-1,2i} + q_{0,2i}  \leq  \pi_{2i},
\end{gathered}
\end{equation}
the constraints (\ref{eq:corner-constraints-star-interior}) and (\ref{eq:corner-constraints-star-non-interior}) are obtained,
where
(\ref{eq:corner-constraints-star-interior})
is valid for $\Pi \leq 2 \pi_{0}$
and
(\ref{eq:corner-constraints-star-non-interior})
is valid for $\Pi > 2 \pi_{0}$.
This results in the conclusion stated in Lemma \ref{lemma:star-high-q-given}, where it is shown that for values of 
$q_{2j-1,2j}$    %
greater than the limits provided in
(\ref{eq:corner-constraints-star-interior}) and (\ref{eq:corner-constraints-star-non-interior}) %
the resultant $SLEM$
is greater than the value of $SLEM$ obtained for
values of 
$q_{2j-1,2j}$    %
that satisfy
(\ref{eq:corner-constraints-star-interior}) and (\ref{eq:corner-constraints-star-non-interior}).
Hence, as stated in Theorem \ref{theorem1209},
for the case of $m \geq 3$, the resultant Pareto frontier is reduced to a single point corresponding to 
$q_{2i-1,2i} = 0$    %
for $i=1,...,m$,    %
which is the minimum point of the optimization problem (\ref{eq:Eq20181027427}).
\subsection{Case of $m=2$}

\subsubsection{Equilibrium Distributions Satisfying (\ref{eq:Eq201810142587})}
\label{sec:m2-with-20}
For a friendship graph with $m=2$, equations (\ref{eq:Eq201807237825}) can be written as below,
\begin{subequations}
	\label{eq:Eq201810132476}
	\begin{align}
		\left( \mu_{i} \left( 1 + \vartheta_{i}^{2} \right) + s - 1 \right) a_{i}
		&=
		-\mu_{i} \vartheta_{1} \vartheta_{2} a_{ (i \; mod \; 2 ) + 1 },
		\label{eq:Eq201810132476a}
		\\
		\left( \mu_{i} \left( 1 + \vartheta_{i}^{2} \right) - s - 1 \right) b_{i}
		&=
		-\mu_{i} \vartheta_{1} \vartheta_{2} b_{ (i \; mod \; 2 ) + 1 },
		\label{eq:Eq201810132476b}
	\end{align}
\end{subequations}
for $i=1,2$.
Considering $a_{1}^{2} = b_{1}^{2}$ and $a_{2}^{2} = b_{2}^{2}$, we have
$( \mu_{i} ( 1 + \vartheta_{i}^{2} ) + s - 1 )^{2}
=
( \mu_{i} ( 1 + \vartheta_{i}^{2} ) - s - 1 )^{2}$,
for $i=1,2$. 
This %
results in either $s=0$ (which is not acceptable) or
\begin{equation}
	\label{eq:Eq201810132524}
	\begin{gathered}
			\mu_{i} 
			= 
			\frac{ 1 }{ 1 + \vartheta_{i}^{2} } 
			= 
			\frac{ \pi_{0} }{   \pi_{0} + \pi_{2i-1} + \pi_{2i} },
			\;\;\; \text{for} \;\;\; i=1,2.
	\end{gathered}
\end{equation}
Hence, (\ref{eq:Eq201810132536}) are obtained for the optimal values of 
$q_{0,2i+j-2}$    %
for $i,j=1,2$.
From (\ref{eq:Eq201810132476a}), we have
$\frac{ a_{1} }{ a_{2} } 
= 
\frac{  -\mu_{1} \vartheta_{1} \vartheta_{2}  }{  \mu_{1} \left( 1 + \vartheta_{1}^{2} \right) + s - 1  } 
= 
\frac{  \mu_{2} \left( 1 + \vartheta_{2}^{2} \right) + s - 1  }{  -\mu_{2} \vartheta_{1} \vartheta_{2}  }
$, 
where 
substituting the optimal value of $\mu_{1}$ and $\mu_{2}$ as provided in (\ref{eq:Eq201810132524}), 
for $s$, we have 
$s^{2} + 1 - \frac{  \vartheta_{1} \vartheta_{2}  }{  \left( 1 + \vartheta_{1}^{2} \right)  \left( 1 + \vartheta_{2}^{2} \right)  } = 0$.
Thus, (\ref{eq:Eq201810132576}) is obtained for the optimal value of $s$ (i.e., $SLEM$).
Since
	$g = 0$,
then from 
$q_{0,1}$ $+$ $q_{0,2}$ $+$ $q_{0,3}$ $+$ $q_{0,4}$ $\leq$ $\pi_{0}$, 
the constraint (\ref{eq:Eq201810142587}) is obtained.
This means that the results in (\ref{eq:Eq201810132536}) and (\ref{eq:Eq201810132576}) are valid for equilibrium distributions that satisfy (\ref{eq:Eq201810142587}).
The optimal values of $SLEM$ obtained in (\ref{eq:Eq201810132576})
is
valid if the single eigenvalues of the transition probability matrix (as provided in (\ref{eq:20181013903})) are smaller than the value of $SLEM$ in absolute value.
This condition imposes the following additional constraints on $q_{(i,1),(i,2)}$,
\begin{equation}
	\label{eq:201810142613}
	\begin{gathered}
				\frac{ \pi_{2i-1} + \pi_{2i} }{ \pi_{0} + \pi_{2i-1} + \pi_{2i} }
			-
			s
			<
			q_{2i-1,2i}    %
			<
			\frac{ \pi_{2i-1} + \pi_{2i} }{ \pi_{0} + \pi_{2i-1} + \pi_{2i} }
			+
			s
	\end{gathered}
\end{equation}
for $i=1,2$, 
where
$s$ is given in (\ref{eq:Eq201810132576}). %
Other constraints for
$q_{1,2}$    %
and
$q_{3,4}$    %
are
$q_{0,2i+j-2} + q_{2i-1,2i} < \pi_{2i+j-2}$, for $i,j=1,2$, 
	which are equivalent to the upper bounds in (\ref{eq:201810142672}).
	Considering the fact that 
	$q_{2i-1,2i}$    %
	for $i=1,2$    %
	are positive, it can be shown that 
	the constraints
	(\ref{eq:201810142613})    %
	apply a lower limit only on one of 
	$q_{2i-1,2i}$ for $i=1,2$,    %
	which depends on 
	the equilibrium distribution satisfying 
	$\pi_{1} + \pi_{2}  >  \pi_{3} + \pi_{4}$,    %
	or its reverse.
	In the special case of equilibrium distributions with 
	$\pi_{1}$ $+$ $\pi_{2}$ $=$ $\pi_{3}$ $+$ $\pi_{4}$, 
	the lower limits in 
	(\ref{eq:201810142613})    %
	are equal to zero and they don't apply any lower limit on 
	$q_{2i-1,2i}$    %
	for $i=1,2$.  %
	As a result, similar to the case of star topology, the optimal value of $s$ given in (\ref{eq:Eq201810132576})
	is valid for all values of 
	$q_{2i-1,2i}$    %
	for $i=1,2$    %
	smaller than
	the upper limits in (\ref{eq:201810142613}). %
	For values of 
	$q_{2i-1,2i}$    %
	for $i=1,2$    %
	greater than
	these upper limits    %
	the optimal value of $SLEM$ increases.
		Thus, Remark \ref{remark:m2-equal-sides} can be concluded. %
	For the equilibrium distributions with 
	$\pi_{1} + \pi_{2} \neq \pi_{3} + \pi_{4}$, 
	the constraints 
	(\ref{eq:201810142613})    %
	impose lower limit on either one of 
	$q_{2i-1,2i}$    %
	for $i=1,2$.
	Such equilibrium distributions result in a Pareto Frontier.
		As an example, in Example \ref{example:1}, we have provided the optimal $SLEM$ for the equilibrium distribution 
		$\pi_{0} = \pi_{4} = 3$,    %
		$\pi_{1} = \pi_{2} = \pi_{3} = 1$.
	\subsubsection{Equilibrium Distributions not Satisfying (\ref{eq:Eq201810142587})}
	\label{sec:m2-without-20}
		For equilibrium distributions that does not satisfy (\ref{eq:Eq201810142587}), the assumption
		$g = 0$  %
		is not true.
		Assuming $a_{1} = a_{2} = 0$ will result in the trivial answer which is not acceptable.
		Thus, we have to assume that $b_{2} = b_{2} = 0$.
		Hence, (\ref{eq:Eq20171114916}) can be written as
			$\left( \left( s - 1 \right) \frac{ a_{i} }{ \mu_i } \right)^{2} = \vartheta_{i}^{2} g^{2} + c_{i}^{2}$
			for $i=1,2$,
		and from the fact that $c_{1}=c_{2}=0$, it can be concluded that
			$\left( \frac{ a_{1} }{  \vartheta_{1} \mu_{1} } \right)^{2} = \left( \frac{ a_{2} }{ \vartheta_{2}\mu_{2} } \right)^{2}$.
		Since
		$g \neq 0$, %
		then we have
		$1 - \vartheta_{1}^{2} \mu_{1} - \vartheta_{2}^{2} \mu_{2} = 0$, where considering
		$\frac{  a_{1}  }{  \vartheta_{1} \mu_{1}  }    =    - \frac{  a_{2}   }{ \vartheta_{2} \mu_{2}  }$,
	the optimal answers are obtained as below,
		\begin{equation}
			\label{eq:Eq201810142663}
			\begin{gathered}
						\mu_{1} = \frac
						{ \pi_0 \left( \pi_0 + 2 \left( \pi_{3} + \pi_{4} \right) \right) }
						{ A_{0} }
						\quad 
						\mu_{2} = \frac
						{ \pi_0 \left( \pi_0 + 2 \left( \pi_{1} + \pi_{2} \right) \right) }
						{ A_{0} }
			\end{gathered}
		\end{equation}
	where 
	$A_{0} = \pi_0 ( \pi_{1} + \pi_{2}  +  \pi_{3} + \pi_{4}  )  +  4 ( \pi_{1} + \pi_{2} ) ( \pi_{3} + \pi_{4} )$.
		Hence, (\ref{eq:Eq201810162962}) are obtained for the optimal values of 
		$q_{0,2i+j-2}$    %
		for $i,j=1,2$.
		and (\ref{eq:Eq201810142678}) is obtained for the optimal value of $s$.

	The optimal values of $SLEM$ obtained in (\ref{eq:Eq201810142678})
	is
	valid if the single eigenvalues of the transition probability matrix (as provided in (\ref{eq:20181013903})) are smaller than the value of $SLEM$ in absolute value.
	This condition imposes the following additional constraints on 
	$q_{2i-1,2i}$,    %
		\begin{equation}
			\label{eq:201810162993}    %
			\begin{gathered}
				B_{i}	\left( 1 - \mu_{i} - s \right)    <    q_{2i-1,2i}    <    B_{i}	\left( 1 - \mu_{i} + s \right)
			\end{gathered}
		\end{equation}
		for $i=1,2$,
	where 
	$\mu_{1}$, $\mu_{2}$ and $s$ are given in
	(\ref{eq:Eq201810142663})    %
	and
	(\ref{eq:Eq201810142678}),    %
	respectively
	and $ B_{i} = \frac{ \pi_{2i-1} \pi_{2i} }{ \pi_{2i-1} + \pi_{2i} } $.
	Other constraints for 
	$q_{1,2}$    %
	and 
	$q_{3,4}$    %
	are
	$q_{0,2i+j-2} + q_{2i-1,2i} < \pi_{2i+j-2}$, for $i,j=1,2$,    %
			which are equivalent to (\ref{eq:201810163031b}) and (\ref{eq:201810163031d}).
		Substituting the values of $\mu_{1}$, $\mu_{2}$ and $s$ from
		(\ref{eq:Eq201810142663})    %
		and
		(\ref{eq:Eq201810142678}),    %
		in the lower limits
			of (\ref{eq:201810162993}),
			the lower bounds (\ref{eq:201810163031a}) and (\ref{eq:201810163031c}) are obtained.
				Similar to Subsection \ref{sec:m2-with-20}, based on positivity of 
				$q_{2i-1,2i}$    %
				for $i=1,2$, 
				it can be shown that 
				the constraints 
				(\ref{eq:201810162993}) 
				apply a lower limit only on one of 
				$q_{2i-1,2i}$    %
				for $i=1,2$, %
				and 
				Remark \ref{remark:m2-equal-sides} can be concluded.
				In Example \ref{example:2}, we have provided the optimal $SLEM$ for an example with 
				$\pi_{1} + \pi_{2} \neq \pi_{3} + \pi_{4}$,    %
				which result in a Pareto Frontier.
							\subsection{Case of $m=1$}
							As mentioned in Subsection \ref{sec:main-results-m1},
							for the friendship graph with $m=1$,
							the graph is reduced to a triangle where we denote the central vertex by index $3$ and the other two vertices by indices $1$ and $2$.
							Note that $q_{1,2}$ is given and
							instead of solving  the optimization problem (\ref{eq:Eq201811023654}), we address the problem by solving the optimization problem (\ref{eq:Eq201712231257}), via SDP.
							In doing so, the matrix $\boldsymbol{L}(q)$, can be written as below,
								$\boldsymbol{D}^{-\frac{1}{2}}  \boldsymbol{L}(q)  \boldsymbol{D}^{-\frac{1}{2}}
								=$ 
								$q_{1,3} 
								\boldsymbol{e}_{3,1} \boldsymbol{e}_{3,1}^{T}
								+
								q_{2,3} 
								\boldsymbol{e}_{3,2} \boldsymbol{e}_{3,2}^{T}
								+
								q_{1,2} 
								\boldsymbol{e}_{1,2} \boldsymbol{e}_{1,2}^{T}
								$
							where
								$\boldsymbol{e}_{i,j} = \frac{ \boldsymbol{e}_{i} }{ \sqrt{\pi_{i}} } - \frac{ \boldsymbol{e}_{j} }{ \sqrt{\pi_{j}} }$,
								and 
								$\boldsymbol{e}_{1}^{T} = \left[ 1, 0, 0 \right]$, 
								$\boldsymbol{e}_{2}^{T} = \left[ 0, 1, 0 \right]$, 
								and
								$\boldsymbol{e}_{3}^{T} = \left[ 0, 0, 1 \right]$.
							The 
							optimization problem (\ref{eq:Eq201712231257}) can be formulated as the following semidefinite programming problem  \cite{BoydConvexBook,JafarizadehIEEESensors2011},
								\begin{subequations}
									\label{eq:Eq2018110311249}
									\begin{align}
										\min_{\substack{ \{ q_{1,3}, q_{2,3}, s \} }}
										\;\;
										&s,
										\nonumber
										\\
										s.t.
										\quad
										&
											(s - 1) \boldsymbol{I}
											+ \boldsymbol{D}^{-\frac{1}{2}} \boldsymbol{L}(q) \boldsymbol{D}^{-\frac{1}{2}}
											+ \widetilde{\boldsymbol{J}}
											\succcurlyeq    %
											0
										\label{eq:Eq2018110311249a}
										\\
										&
											(s + 1) \boldsymbol{I}
											- \boldsymbol{D}^{-\frac{1}{2}} \boldsymbol{L}(q) \boldsymbol{D}^{-\frac{1}{2}}
											- \widetilde{\boldsymbol{J}}
											\succcurlyeq    %
											0
										\label{eq:Eq2018110311249b}
										\\& 
											\pi_{1}  -   q_{1,2}   -  q_{1,3}   \geq 0,
											\; %
											\pi_{2}  -  q_{1,2}  -  q_{2,3}  \geq 0, 
											\; %
											\pi_{3}  -  q_{1,3}   -  q_{2,3}   \geq 0
									\end{align}
								\end{subequations}
							where
							$\widetilde{\boldsymbol{J}} = \frac{  \boldsymbol{D}^{\frac{1}{2}} \boldsymbol{1} \boldsymbol{1}^{T} \boldsymbol{D}^{\frac{1}{2}}  }{  \sum_{i=1}^{3} \pi_{i}  }$.
							Introducing
							$\boldsymbol{x} = [ q_{1,3}, q_{2,3}, s]^{T}$,
							$\boldsymbol{c} = [ 0, 0, 1]^{T}$,
							and the block diagonal matrices
							$\boldsymbol{F}_{0}$, $\boldsymbol{F}_{1}$, $\boldsymbol{F}_{2}$ and $\boldsymbol{F}_{s}$
							as %
							$\boldsymbol{F}_{0}$  $=$  $diag($
							$\widetilde{\boldsymbol{J}}$ $-$ $\boldsymbol{I}$ $-$ $\boldsymbol{\zeta}_{1,2}$ $\boldsymbol{\zeta}_{1,2}^{T}$, 
							$\boldsymbol{I}$ $+$ $\boldsymbol{\zeta}_{1,2}$ $\boldsymbol{\zeta}_{1,2}^{T}$ $-$ $\widetilde{\boldsymbol{J}}$, 
							$\pi_{1}$ $-$ $q_{1,2}$, 
							$\pi_{2}$ $-$ $q_{1,2}$, 
							$\pi_{3}$ 
							$)$, 
							$\boldsymbol{F}_{1}$ $=$
							$diag($
							$\boldsymbol{\zeta}_{1,3}$ $\boldsymbol{\zeta}_{1,3}^{T}$,
							$-$ $\boldsymbol{\zeta}_{1,3}$ $\boldsymbol{\zeta}_{1,3}^{T}$,
							$-$ $1$,
							$0$,
							$-1$
							$)$,
							$\boldsymbol{F}_{2}$ $=$
							$diag($
							$\boldsymbol{\zeta}_{2,3}$  $\boldsymbol{\zeta}_{2,3}^{T}$,
							$-$ $\boldsymbol{\zeta}_{2,3}$  $\boldsymbol{\zeta}_{2,3}^{T}$,
							$0$,
							$-1$,
							$-1$
							$)$,
							$\boldsymbol{F}_{s}$ $=$ $diag($ $\boldsymbol{I}$, $\boldsymbol{I}$, $0$, $0$, $0$ $)$
							where
							$\boldsymbol{\zeta}_{i,j}    =    \left(  \frac{ \boldsymbol{e}_{i} }{ \sqrt{ \pi_{i} } }  -  \frac{ \boldsymbol{e}_{j} }{ \sqrt{ \pi_{j} } }  \right)$,
							problem (\ref{eq:Eq201801262501}) can be written in the standard form of the semidefinite programming \cite{BoydConvexBook,JafarizadehIEEESensors2011}
							as provided in (\ref{eq:Eq201801272653}), and
							the dual problem is as in (\ref{eq:Eq201801272673}).
							The dual variable $\boldsymbol{Z}$ can be written as
							$\boldsymbol{Z}$  $=$
							$[ \boldsymbol{Z}_{1}^{T}$, $\boldsymbol{Z}_{2}^{T}$, $\boldsymbol{Z}_{3}^{T} ]^{T}$,
							where
							the variables
							$\boldsymbol{Z}_{1}$, $\boldsymbol{Z}_{2}$, $\boldsymbol{Z}_{3}$
							are as below,
								\begin{equation}
									\label{eq:20180914989}
									\begin{gathered}
										\boldsymbol{Z}_{1}
										=
										a_{1,3}
										\widetilde{ \boldsymbol{e}_{3,1} }
										+
										a_{2,3}
										\widetilde{ \boldsymbol{e}_{3,2} },
										\;\; \;\; %
										\boldsymbol{Z}_{2}
										=
										b_{1,3}
										\widetilde{ \boldsymbol{e}_{3,1} }
										+
										b_{2,3}
										\widetilde{ \boldsymbol{e}_{3,2} },
										\;\;\;\; %
										\boldsymbol{Z}_{3}
										=
										c_{1} \boldsymbol{e}_{1}
										+
										c_{2} \boldsymbol{e}_{2}
										+
										c_{3} \boldsymbol{e}_{3}.
									\end{gathered}
								\end{equation}
								$\widetilde{ \boldsymbol{e}_{3,i} }$
								for $i = 1, 2$ is the dual of
								$\boldsymbol{e}_{3,i}$
								such that
								$\boldsymbol{e}_{3,i}^{T}     \widetilde{ \boldsymbol{e}_{3,j} } = \delta_{i,j}$,
							where $\delta_{i,j}$ is the Kronecker delta function.
							Considering the relation
								$\boldsymbol{e}_{1,2}    =    \boldsymbol{e}_{3,2}    -    \boldsymbol{e}_{3,1}$, 
							The complementary slackness condition \cite{BoydConvexBook,JafarizadehIEEESensors2011}, 
								can be written as below,
								\begin{subequations}
									\label{eq:201809141080}
									\begin{gather}
											\left(
											q_{i,3} \left( \frac{ 1 }{ \pi_{i} } + \frac{ 1 }{ \pi_{3} } \right)
											+
											\frac{ q_{1,2} }{ \pi_{i} }
											+
											s - 1
											\right)
											a_{i,3}
											=
											\left( -\frac{ q_{i^{'}, 3} }{ \pi_{3} }+ \frac{ q_{1,2} }{ \pi_{i} } \right)
											a_{i^{'}, 3}
										\label{eq:201809141080a}
										\\
											\left(
											q_{i,3} \left( \frac{ 1 }{ \pi_{i} } + \frac{ 1 }{ \pi_{3} } \right)
											+
											\frac{ q_{1,2} }{ \pi_{i} }
											- s - 1
											\right)
											b_{i,3}
											=
											\left( - \frac{ q_{i^{'}, 3} }{ \pi_{3} }   +  \frac{ q_{1,2} }{ \pi_{i} } \right) b_{i^{'}, 3}
										\label{eq:201809141080b}
									\end{gather}
								\end{subequations}
								for $i=1,2$, 
								where $i^{'} = ( i \; mod \; 2) + 1$.
							Using (\ref{eq:20180914989}), dual constraints, %
							$tr\left[ \boldsymbol{Z} \boldsymbol{F}_{i} \boldsymbol{Z}^{T} \right] = c_{i} = 0$,
							can be written as below,
								\begin{equation}
									\label{eq:201809292663}
									\begin{gathered}
											a_{1,3}^{2}
											-
											b_{1,3}^{2}
											-
											\left(  c_{3}^{2}  +  c_{1}^{2}  \right)
											=  0,
										\;\; \;\; %
											a_{2,3}^{2}
											-
											b_{2,3}^{2}
											-
											\left(  c_{3}^{2}  +  c_{2}^{2}  \right)
											=  0.
									\end{gathered}
								\end{equation}
								Assuming
								$q_{1,3} + q_{2,3} < \pi_{3}$,
								$q_{1,2} + q_{2,3} < \pi_{2}$,
								and
								$q_{1,2} + q_{1,3} < \pi_{1}$,
								we have
								$c_{1} = c_{2} = c_{3} = 0$.
								Thus, from (\ref{eq:201809292663}), it can be concluded that
									$a_{1,3}^{2}	= 	b_{1,3}^{2}$, 
									and 
									$a_{2,3}^{2}  =  b_{2,3}^{2}$,
									where considering
								(\ref{eq:201809141080}) for $i=1$,
							it can be concluded that
								\begin{equation}
									\label{eq:201809151178}
									\begin{gathered}
											q_{1,3} \left(  \frac{ 1 }{ \pi_{1} }  +  \frac{ 1 }{ \pi_{3} }  \right)
											+
											\frac{ q_{1,2} }{ \pi_{1} }
											= 1,
											\;\; \;\; %
											q_{2,3} \left(  \frac{ 1 }{ \pi_{2} }  +  \frac{ 1 }{ \pi_{3} }  \right)
											+
											\frac{ q_{1,2} }{ \pi_{2} }
											= 1.
									\end{gathered}
								\end{equation}
								From (\ref{eq:201809151178}), the results in (\ref{eq:201809151209-main-result}) can be concluded for $q_{1,3}$ and $q_{2,3}$.
								Substituting (\ref{eq:201809151178}) in (\ref{eq:201809141080a}) for $i=1,2$,
							we have
							$s^{2}$
							$=$
							$(  \frac{ q_{2,3} }{ \pi_{3} } - \frac{ q_{1,2} }{ \pi_{1} }  )(  \frac{ q_{1,3} }{ \pi_{3} }  -  \frac{ q_{1,2} }{ \pi_{2} }  )$,
								where substituting %
								(\ref{eq:201809151209-main-result}),    %
								the result in (\ref{eq:201809151244-main-result}) is obtained for $s$.
							Considering the fact that 
								the constraint $q_{1,2} \leq \min\{ \pi_{1}, \pi_{2} \}$ is always satisfied,
							selecting $q_{1,3}$ and $q_{2,3}$ as given in
							(\ref{eq:201809151209-main-result}),    %
							guarantees that the constraints
							$q_{1,2} + q_{2,3} \leq \pi_{2}$,
							and
							$q_{1,2} + q_{1,3} \leq \pi_{1}$
							hold true.
							Substituting the values of $q_{1,3}$ and $q_{2,3}$ in 
							constraint $q_{1,3}$ $+$ $q_{2,3}$ $\leq$ $\pi_{3}$, 
							the lower limit (\ref{eq:201809302893-main-result}) can be concluded for $q_{1,2}$.
								For $\pi_{3}^{2} > \pi_{1} \pi_{2}$, the lower bound provided in 
								(\ref{eq:201809302893-main-result})    %
								is negative and it does not impose any lower limit on $q_{1,2}$.
								Thus, for equilibrium distributions with $\pi_{3}^{2} \geq \pi_{1} \pi_{2}$, the optimal weights and $SLEM$ in
								(\ref{eq:201809151209-main-result})    %
								and
								(\ref{eq:201809151244-main-result})    %
								hold true for all values of $q_{1,2}$.
								It is obvious that the Pareto frontier for the friendship graph with $m=1$ and equilibrium distribution $\pi_{3}^{2} > \pi_{1} \pi_{2}$ is the line
								$s =$ 
								$( \pi_{1}$ $\pi_{2}$ $-$ $($ $\pi_{1}$ $+$ $\pi_{2}$ $+$ $\pi_{3} )$ $q_{1,2}  )$
								$/$
								$\sqrt{ \pi_{1} \pi_{2} \left( \pi_{1} + \pi_{3} \right)  \left( \pi_{2} + \pi_{3} \right) }$
								for
								$  0  \leq  q_{1,2}  \leq  \frac{  \pi_{1} \pi_{2}  }{  \pi_{1} + \pi_{2} + \pi_{3}  }  $.
							\subsubsection{Case of $\pi_{3}^{2} < \pi_{1}\pi_{2}$}
							For $q_{12}$ satisfying 
							(\ref{eq:201809302893-main-result}),    %
							the optimal weights and the $SLEM$ are same as (\ref{eq:201809151209-main-result}) and (\ref{eq:201809151244-main-result}).
							In the case of $q_{12}$ smaller than the limit provided in 
							(\ref{eq:201809302893-main-result}),    %
							the optimal results can be obtained for two different ranges of $q_{1,2}$.
							In the case of the region with lower limit of zero, i.e., $q_{1,2} > 0$, we have
							$b_{2,3} = b_{1,3} = 0$,
							$c_{1} = c_{2} = 0$, $c_{3} \neq 0$,
							where 
							it can be concluded that $a_{1,3}^{2}  =  a_{2,3}^{2}  =  c_{3}^{2}$.
							Also, from $c_{3} \neq 0$, we have
							$q_{1,3}  +  q_{2,3}  =  \pi_{3}$. 
							From $a_{1,3}^{2}  =  a_{2,3}^{2}  =  c_{3}^{2}$, 
							either $a_{1,3} = a_{2,3}$ or $a_{1,3} = -a_{2,3}$ can be derived.
							$a_{1,3} = a_{2,3}$ results in unacceptable $SLEM$, while on the other hand
							for $a_{1,3} = -a_{2,3}$, the following results are obtained
							\begin{subequations}
								\label{eq:201809151364}
								\begin{gather}
									q_{1,3}
									=
									\pi_{1} \pi_{3} ( 2\pi_{2} + \pi_{3} + 2q_{1,2} ) - 2 \pi_{2} \pi_{3} q_{1,2} / B_{0},
									\label{eq:201809151364a}
									\\
									q_{2,3}
									=
									\left( \pi_2\pi_3(2\pi_1 + \pi_3 + 2q_{1,2}) - 2\pi_1\pi_3q_{1,2} \right) / B_{0},
									\label{eq:201809151364b}
								\end{gather}
							\end{subequations}
							\begin{equation}
								\label{eq:201809151355}
								\begin{gathered}
									s = 
									\left( 4\pi_1\pi_2 -\pi_3^2 -4q_{1,2}(\pi_1+\pi_2+\pi_3 ) \right) / B_{0}.
								\end{gathered}
							\end{equation}
							where $B_{0} = 4 \pi_{1} \pi_{2} + \pi_{3} \left( \pi_{1} + \pi_{2} \right)$. 
							Note that the results in (\ref{eq:201809151364}) and (\ref{eq:201809151355}) hold true for equilibrium distributions with $q_{1,2} \geq 0$.
							But, for the case of $\pi_{1} \neq \pi_{2}$ providing the closed-form formula for the upper limit of this region is cumbersome and we omit presenting it for the general case.
							In Example \ref{example:3}, we address the case of $\pi_{1} = 2$, $\pi_{2} = \pi_{3} = 1$, where all three ranges of $q_{12}$ are present.
							In Subsection \ref{sec:pi1=pi2},
							we have address %
							case of $\pi_{1} = \pi_{2}$, where the three ranges are reduced to two
							and we have provided closed-form formulas for the optimal weights and $SLEM$.

					\section{Conclusions}
					\label{sec:Conclusions}
					Arising from the information diffusion model developed based on DeGroot's consensus model, this paper addresses the FMRMC problem over the friendship graph with given transition probabilities on a subset of edges among friends.
					For the resultant multi-objective optimization problem, the corresponding Pareto frontier has been provided for the friendship graph with different number of blades. 
					Interestingly, for friendship graphs with more than two blades, the Pareto frontier is reduced to a single minimum optimal point, which is same as the optimal point corresponding to the minimum spanning tree of the friendship graph, i.e., the star topology.
					Based on this result, it has been shown that there is a lower limit for the transition probabilities on among friends, where for values higher than this limit the transition probabilities do not have any impact on the convergence rate. 
					One possible direction for future work is to investigate other topologies where the Pareto frontier is reduced to the optimal point corresponding to its minimum spanning tree.

					\section*{Acknowledgements}
					The author would like to thank Mr. Satoshi Egi and Dr. Yusaku Kaneta of 
					Rakuten Institute of Technology, %
					for the valuable 
					discussions    %
					regarding the paper.

\appendix

\section{Proof of Lemma \ref{sec:Lemma-Star_Friendship-Eigenstructure}}
\label{sec:Appendix-Lemma-Proof-Star_Friendship-Eigenstructure}
Let the row vector $\boldsymbol{V}$ be the right eigenvector of
$\boldsymbol{P}_{S}$    %
corresponding to eigenvalue $\lambda$.
Thus, we have
$\boldsymbol{P}_{S} \times \boldsymbol{V}    =    \lambda \boldsymbol{V}$.
Multiplying both sides of 
this equation %
from the left hand side with $\boldsymbol{\Lambda}$, we have
$\boldsymbol{\Lambda} \times \boldsymbol{P}_{S} \times \boldsymbol{V}    =    \lambda \boldsymbol{\Lambda} \times \boldsymbol{V}$
where from (\ref{eq:20181009332}), we have
$\boldsymbol{P} \times \boldsymbol{\Lambda} \times \boldsymbol{V}    =    \lambda \boldsymbol{\Lambda} \times  \boldsymbol{V}$.
Thus, it can be concluded that $\boldsymbol{\Lambda} \times \boldsymbol{V}$ is the right eigenvector of
$\boldsymbol{P}$  %
corresponding to eigenvalue $\lambda$.
The rank of $\boldsymbol{\Lambda}$ is equal to $m+1$.
Therefore, $\boldsymbol{\Lambda} \times \boldsymbol{V}$ is equal to $\boldsymbol{0}$ only if $\boldsymbol{V}$ is equal to zero.
Hence, it can be concluded that every eigenvalue of
$\boldsymbol{P}_{S}$    %
is also an eigenvalue of
$\boldsymbol{P}$.    %

	\section{Proof of Lemma \ref{lemma:star-high-q-given}}
	\label{lemma:star-high-q-given-proof}
		Without loss of generality, let 
		$\pi_{2j-1} \leq \pi_{2j}$.   %
		Assuming that the value of 
		$q_{2j-1,2j}$    %
		does not satisfy the constraints
		(\ref{eq:corner-constraints-star-interior}) and (\ref{eq:corner-constraints-star-non-interior}), %
		the constraints (\ref{eq:corner-constraints-star-raw}) reduce to equality,
		where the optimal value of 
		$q_{0,2j-1}$    %
		can be written as %
		$q_{0,2j-1}  =  \pi_{2j-1}  -  q_{2j-1,2j}$. 
		Since the value of 
		$q_{0,2j-1}$    %
		is different than the optimal ones provided in
		(\ref{eq:Eq201807237981}) %
		and
		(\ref{eq:Eq201807278235}) %
		then the resultant $SLEM$ is greater than the one obtained using
		(\ref{eq:Eq201807237981}) %
		and
		(\ref{eq:Eq201807278235}). %
		Thus, Lemma \ref{lemma:star-high-q-given} can be concluded.
\section{Interlacing \cite{HAEMERS1995593}}
\label{sec:AppendixInterlacing}
Consider two sequences of real numbers:
$\lambda_{1} \geq ... \geq \lambda_{n}$
and
$\mu_{1} \geq ... \geq \mu_{m}$,
with $m < n$.
The second sequence is said to interlace the first one whenever
$ \lambda_{i} \geq \mu_{i} \geq \lambda_{n-m+i} \;\; \text{for} \;\; i=1,...,m$.
The interlacing is called tight if there exist an integer $k \in [0,m]$ such that
$ \lambda_{i} = \mu_{i} \;\; \text{for} \;\; 1 \leq i \leq k \quad \text{and} \quad \lambda_{n-m+i} = \mu_{i} \;\; \text{for} \;\; k+1 \leq i \leq m$.
\bibliography{ArXiv_Format}

\end{document}